%% file: outageBehavOfRandPrecIF_NoTimeExtPaper_SphereSurfaceArea1UriG.tex
\documentclass[article,twocolumn]{IEEEtran}

\makeatother
\usepackage{amssymb,amsmath}
\usepackage{bbding}
\usepackage{flushend,cuted}
\usepackage{amsthm}
\usepackage{amsfonts}
\usepackage{float} 
\usepackage{srcltx}
\usepackage{mathrsfs} 
\usepackage{multirow}
\usepackage{amssymb}

\usepackage{graphicx}
\usepackage{epsfig}
\usepackage{psfrag}
\usepackage{subfigure}

\usepackage{setspace}
\usepackage{multicol}

\usepackage{stfloats}

\usepackage[]{units} 
\usepackage{url} 
\usepackage[dvips]{color}
\usepackage{verbatim} 
\usepackage{cite} 
\usepackage{ifthen} 
\usepackage{ifpdf} 
\usepackage{soul}
\usepackage{mathtools}

\makeatletter
\def\widebreve#1{\mathop{\vbox{\m@th\ialign{##\crcr\noalign{\kern3\p@}%
      \brevefill\crcr\noalign{\kern3\p@\nointerlineskip}%
      $\hfil\displaystyle{#1}\hfil$\crcr}}}\limits}

\def\brevefill{$\m@th \setbox\z@\hbox{$\braceld$}%
  \bracelu\leaders\vrule \@height\ht\z@ \@depth\z@\hfill\braceru$}
\makeatletter

\input{eli_macros.tex}
\input{rv_defs.tex}

\newtheorem{lemma}{Lemma}
\newtheorem{corollary}{Corollary}
\newtheorem{theorem}{Theorem}
\newtheorem{definition}{Definition}

\newtheorem{remark}{Remark}

\newcommand{\SNR}{\text{$\mathsf{SNR}$}}
\DeclareRobustCommand{\prob}[1][{\rm Pr}]{\ensuremath {{#1}}}
\DeclareRobustCommand{\aNorm}[1][aNorm]{\ensuremath {\|{\bf a}\|}}
\DeclareRobustCommand{\Nt}[1][Nt]{\ensuremath {N_t}}
\DeclareRobustCommand{\alNt}[1][Nt]{\alpha(\Nt)}
\DeclareRobustCommand{\genGam}[1][\beta]{\ensuremath {{#1}}}

\usepackage{fancyhdr}
\begin{document}
\allowdisplaybreaks
\title{Outage Behavior of Integer Forcing With Random Unitary Pre-Processing}

\author{Elad Domanovitz and Uri Erez
\thanks{The work of E. Domanovitz and U. Erez was supported in part by the Israel Science Foundation under Grant No. 1956/16 and
 by the Heron consortium via the Israel Ministry of Economy
and Industry.}
\thanks{The material in this paper was presented in part at the 2016 IEEE
International Symposium on Information Theory, Barcelona.}
\thanks{E. Domanovitz and U. Erez are with the Department of Electrical Engineering -- Systems, Tel Aviv University, Tel Aviv, Israel (email: domanovi,uri@eng.tau.ac.il).}
}
\maketitle


\begin{abstract}
Integer forcing is an equalization scheme for the multiple-input multiple-output
communication channel that has been demonstrated to allow operating close to capacity
for ``most" channels. In this work, the measure of ``bad" channels is quantified by
considering a compound channel setting where the transmitter communicates over a fixed
channel but knows only its mutual information. The transmitter encodes the data into independent
streams, all taken from the same linear code. The coded streams are transmitted after applying a unitary transformation. At the receiver side, integer-forcing equalization is applied, followed by standard
single-stream decoding.
Considering pre-processing matrices drawn from a random ensemble,
outage corresponds to the event that the target rate exceeds the achievable rate of integer forcing for a given channel matrix.
For the case of the circular unitary ensemble, an explicit universal bound on the outage probability for a given target rate is derived that holds
for any channel in the compound class.
The derived bound depends only on the gap-to-capacity and the number of transmit antennas.
The results are also applied to obtain universal bounds on the gap-to-capacity
of multiple-antenna closed-loop multicast, achievable via linear pre-processed integer forcing.

\end{abstract}

%
\section{Introduction}
%
\label{sec:intro}

The Multiple-Input Multiple-Output (MIMO) Gaussian channel is central to modern communication
and has been extensively studied over the past several decades. Nonetheless, while the capacity limits, under different assumptions on the availability of channel state information, are well understood, the design of low-complexity communication schemes that approach these limits still poses challenges in some scenarios.

For a static channel and  a point-to-point closed-loop setting, capacity may be approached without much difficulty by employing an architecture that decouples coding and modulation.
That is, one may use ``off-the-shelf" codes in conjunction with linear pre- and post-processing based on matrix decompositions.
For instance, one may  use the singular-value decomposition (SVD) to transform the channel into parallel scalar additive white Gaussian noise (AWGN) channels, over which
standard codes may be employed\cite{telatar}. Alternatively, standard scalar codes may be  used in conjunction with the $QR$  decomposition and successive interference cancellation (SIC), see, e.g., \cite{tse2005fundamentals}.
Coding for MIMO channels in an ergodic fading environment is more involved but has also been successfully addressed. See, e.g., \cite{AchievingNearCapacityOnAmultipleAntennaChannel:HochwaldtenBrink2003}.

In contrast, we address the problem of coding over a compound MIMO channel.
More specifically, the focus of this paper is on static (and frequency-flat)  MIMO channels where the transmitter only knows (or may only utilize its knowledge of) the mutual information of the channel.


The design of a practical coding scheme for such a compound MIMO channel scenario was addressed in \cite{OrdentlichErez:IFUniversallyAchievesCapacityUpToGap:2013} where an architecture employing
space-time linear pre-processing (that is independent of the channel) at the transmitter side and integer-forcing (IF) equalization at the receiver side was proposed.
It was shown that such an architecture \emph{universally} achieves the MIMO capacity up to a constant gap, provided the space-time pre-processing satisfies the non-vanishing determinant (NVD) criterion \cite{ExplicitSpaceTimeCodesAchievingtheDiversity:Elia2006}.
While this result is encouraging as it points to the robustness of the IF scheme, the derived gap is very large and calls for further work.

In the present work, we  study the performance of IF where random unitary linear pre-processing is performed over
the spatial dimension only. 
Rather than aiming at guaranteeing successful transmission, we study the outage probability of the scheme.\footnote{This approach is similar to that taken in \cite{CoRa:Larsson2004} and \cite{TransmitDiversityOverQuasiStaticFadingChannelsUsingMultipleAntennasAndRandomSignalMapping:Yingxue2003} with respect to other transmission schemes.}
We focus attention to pre-processing matrices drawn from the isotropic (circular) unitary ensemble as this ensures that all channels having the same singular values, will have the same outage probability.

It is worth noting that the random pre-processing operation serves the purpose of quantifying the measure of ``bad" channels for IF receivers. We further note that the receiver considered is the standard one of \cite{IntegerForcing} but whereas
the results of \cite{IntegerForcing} deal with distributed transmit antennas (i.e., with no encoding across the transmit antennas), in the present paper joint (unitary) pre-processing
is assumed. While this may appear to preclude a distributed setting, some important statistical scenarios are in fact covered by the model. Specifically, in the case of a channel matrix whose entries are i.i.d. Gaussian random variables,
the random unitary transformation assumed in the analysis to follow is in reality
performed by nature, as discussed in Section~\ref{sec:lin-pre} below.


The outage probability in the considered setting thus corresponds to a scheme outage.\footnote{We use the term ``scheme outage" as opposed to ``channel outage." Specifically, in the considered setting, the channel is known to have sufficient mutual information to support the chosen target transmission rate.}
Namely, it is the probability that a random linear pre-processing matrix results in an effective channel for which the rate achievable with an IF receiver is smaller than the target rate. In order to provide universal performance guarantees, we study the \emph{worst-case} outage probability with respect to all possible singular value combinations corresponding to a given mutual information. Thus, the guaranteed performance does not depend on channel statistics.

We note that the performance of a
coding scheme over the compound channel is a strong measure of
its robustness. Clearly, performance guarantees for the compound channel immediately translate to guarantees for a statistical
channel model (as explained in the next section).
To the best of our knowledge, IF is the first practical scheme
for which (provable) universal bounds are known for the MIMO
channel.

We begin by empirically observing that space-only linear pre-processed IF (P-IF) has  greatly improved  performance, in terms of worst-case outage probability, compared to standard linear equalization.
We then derive an explicit bound on the performance of P-IF that depends only on the number of transmit antennas and the gap-to-capacity, where moderate
gaps suffice to guarantee a small outage probability.

As another example of an application of the results, we use the probabilistic method to obtain guarantees on the number of users that can be supported in closed-loop MIMO multicast (guaranteeing no outage occurs) as a function of the gap-to-capacity, when using linear pre-processed IF.

%

The paper is organized as follows. Section~\ref{sec:channel_model} defines the channel model of interest and formulates the problem described above. Section~\ref{sec:IFequalizationSec} provides  background on the integer-forcing
receiver as well as its use in conjunction with linear pre-processing.
Section~\ref{sec:BoundforNrxNtchannels} derives a universal upper bound for the outage probability of randomly linear pre-processed IF over the compound MIMO channel; tighter bounds for the specific case of two transmit antennas and a receiver employing a successive interference cancellation (SIC) variant of IF are also derived. Section~\ref{sec:BoundForClosedLoopMulticast} describes the application of the derived bounds to a close-loop MIMO multicast setting.

\section{Channel Model and Problem Formulation}
\label{sec:channel_model}
A point-to-point (complex) MIMO channel is considered where the transmitter is equipped with $\Nt$ antennas and the receiver is equipped with an arbitrary number ($N_r$) of antennas. Thus, a channel is described  by the relation
\begin{align}
\boldsymbol{y}_c=\svv{H}_c\boldsymbol{x}_c+\boldsymbol{z}_c,
\label{eq:channel_model}
\end{align}
where $\svv{x}_c\in \mathbb{C}^{\Nt}$ is the channel input vector, $\svv{y}_c \in \mathbb{C}^{N_r}$ is the channel output vector, $\svv{H}_c$ is an $N_r\times \Nt$ complex channel matrix, and
$\boldsymbol{z}_c$ is an additive noise vector of  i.i.d. unit variance circularly-symmetric complex Gaussian random variables.\footnote{We denote all complex variables with $c$ to distinguish them from their real-valued representation.}
The input vector $\boldsymbol{x}_c$ is subject to the power constraint\footnote{We denote by $[\cdot]^T$, transpose of a vector/matrix and by $[\cdot]^H$, the Hermitian transpose of a vector/matrix.}
\begin{eqnarray}
\mathbb{E}(\boldsymbol{x}_c^H\boldsymbol{x}_c)\leq \Nt \cdot \SNR.
\end{eqnarray}
We assume that the channel is fixed throughout the whole transmission of a codeword.

For a given input covariance matrix $\svv{Q}_c$, satisfying the power constraint $\Tr(\svv{Q}_c) \leq \Nt · \SNR$, the mutual information of the channel (\ref{eq:channel_model}) is maximized by a Gaussian input, and is given by
\begin{align}
C=\log\det \left( \svv{I}_{N_r\times N_r}+\svv{H}_c\svv{Q}_c\svv{H}_c^H \right). 
\label{eq:cl_capacity}
\end{align}

When it comes to designing transmission strategies, without loss of generality we may assume that $\svv{Q}=\svv{I}$ (isotropic transmission). Namely, we may ``absorb'' $\svv{Q}$ into the channel matrix by replacing $\svv{H}_c$ in~(\ref{eq:cl_capacity}) with $\bar{\svv{H}}_c=\svv{H}_c\svv{Q}^{1/2}$ (and with abuse of notation, we omit the bar). Similarly, we may set $\SNR=1$.
Hence, we may rewrite \eqref{eq:cl_capacity} as
\begin{align}
C&=\log \det \left(\svv{I}_{{N_r}\times{N_r}}+ \svv{H}_c\svv{H}_c^H\right) \nonumber \\
&=\log \det \left(\svv{I}_{N_t\times N_t}+ \svv{H}_c^H\svv{H}_c\right).
\end{align}
We define the set
\begin{align}
\mathbb{H}(C;N_t) = \left\{ \svv{H}_c : \log \det \left(\svv{I}_{N_t\times N_t} +  \svv{H}_c^H \svv{H}_c\right) = C \right\}
\label{eq:setOfChannels}
\end{align}
of all channel matrices $\svv{H}_c$ with $N_t$ transmit antennas and an \emph{arbitrary} number of receive antennas, having the same WI mutual information $C$.

The corresponding compound channel model is defined by (\ref{eq:channel_model}) with the channel matrix $\svv{H}_c$ arbitrarily
chosen from the set $\mathbb{H}(C;N_t)$. The matrix $\svv{H}_{c}$ that was chosen by nature is revealed to the receiver, but not to the
transmitter. Clearly, the capacity of this compound channel is $C$, and is achieved with an isotropic Gaussian input.

We note that the assumption that $N_r$ is arbitrary (i.e., universality) comes at a price. Specifically, restricting the number of receive antennas to a fixed number (more specifically to a value $N_r<N_t$) may be leveraged to obtain
improved performance and bounds since this amounts to  limiting the set over which we take the worst-case channel.
Nonetheless, as we will see, integer forcing behaves
well even in  the considered universal setting.


Employing the IF receiver allows approaching $C$ for ``most'' but not all matrices $\svv{H}_c\in\mathbb{H}(C;N_t)$. We quantify the  measure of the set of bad channel matrices by considering outage events, i.e., those events
where integer forcing fails even though the channel has sufficient mutual information.
More broadly, for a given coding scheme, denote the achievable rate for a given channel matrix $\svv{H}_c$ as $R_{\rm scheme}(\svv{H}_c)$.
Then, given a target rate $R<C$ and a channel $\svv{H}_c\in\mathbb{H}(C;N_t)$, the scheme is in outage when $R_{\rm scheme}(\svv{H}_c)<R$.
For the case of integer forcing, the explicit expression for $R_{\rm IF}(\svv{H}_c)$ is recalled in
Section~\ref{sec:IFequalizationSubSecNo}.

Since applying a linear pre-processing matrix $\svv{P}_c$ results in an effective channel $\svv{H}_c \cdot \svv{P}_c$, it follows that the achievable rate of a transmission scheme over this channel is $R_{\rm scheme}(\svv{H}_c\cdot\svv{P}_c)$.
When $\svv{P}_c$ is drawn at random, the latter rate is also random.
The worst-case (WC) scheme outage probability is defined in turn as
\begin{align}
     P^{\rm WC}_{\rm out,scheme}\left(C,R\right)= \sup_{\svv{H}_c\in\mathbb{H}(C;N_t)}\prob\left(R_{\rm scheme}(\svv{H}_c\cdot\svv{P}_c)<R\right),
     \label{eq:P_WC_OUT}
\end{align}
where the probability is over the ensemble of linear pre-processing matrices. The goal of this paper is to quantify the tradeoff between the transmission rate $R$ and the worst-case outage probability of integer forcing  $P^{\rm WC}_{\rm out,IF}\left(C,R\right)$.

\begin{remark}
\label{remark1}
Assume that $\svv{H}_c$ is modeled as having a probability distribution over the compound class $\mathbb{H}(C;N_t)$. In such a case, the outage probability is given by
\begin{align}
\mathbb{E}_{\svv{H}_c}\left[\prob\left(R_{\rm scheme}(\svv{H}_c\cdot\svv{P}_c)<R~|~\svv{H}_c\right)\right]
\label{eq:7}
\end{align}
where the probability is with respect to $\svv{P}_c$ (given $\svv{H}_c$) and the expectation is with respect to the distribution over $\mathbb{H}(C;N_t)$.

Note that in (\ref{eq:P_WC_OUT}), we take the supremum over the entire compound class rather than averaging over a given distribution. Since the average is always smaller than the supremum, it follows that (\ref{eq:P_WC_OUT})  universally upper bounds (\ref{eq:7}). That is, the bound holds for any distribution over $\mathbb{H}(C;N_t)$.

\end{remark}
\begin{remark}
Assume that $\svv{H}_c$ is modeled as having any probability distribution (not restricted to the compound class $\mathbb{H}(C;N_t)$). In such a case, the outage probability can be expressed as
\begin{align}
\mathbb{E}\biggl[\mathbb{E}_{\svv{H}_c}\left[\prob\left(R_{\rm scheme}(\svv{H}_c\cdot\svv{P}_c)<R~|~\svv{H}_c\right)~\Bigl|\Bigr.~C\right]\biggr]
\label{eq:8}
\end{align}
where again, the probability is with respect to $\svv{P}_c$, the inner expectation is with respect to the marginal distribution of $\svv{H}_c$ given the WI-MI $C$, while the outer one is with respect to $C$. Thus,
in order to bound the scheme outage probability, it suffices to know only the distribution of the WI-MI mutual information $C$. Namely, from (\ref{eq:P_WC_OUT}) we have that the outage probability will be no greater than
\begin{align}
\mathbb{E}\left[P^{\rm WC}_{\rm out,scheme}\left(C,R\right)\right]
\end{align}
where the expectation is over $C$.
\end{remark}

\section{Integer-forcing background}
\label{sec:IFequalizationSec}

\subsection{Single-User Integer-Forcing Equalization}
\label{sec:IFequalizationSubSecNo}

\begin{figure*}
\begin{center}
\begin{psfrags}
\psfrag{a1}[][][0.7]{$2\Nt$}
\psfrag{b1}[][][1]{$\boldsymbol{x}_1$}
\psfrag{b2}[][][1]{$\boldsymbol{x}_{2\Nt}$}
\psfrag{h1}[][][1]{$\svv{P}$}
\psfrag{c1}[][][1]{$\svv{H}$}
\psfrag{z1}[][][1]{$z_1$}
\psfrag{z2}[][][1]{$z_{2N_r}$}
\psfrag{y1}[][][1]{$y_1$}
\psfrag{y2}[][][1]{$y_{2N_r}$}
\psfrag{v1}[][][1]{$v_1$}
\psfrag{v2}[][][1]{$v_{2N_t}$}
\psfrag{e1}[][][1]{$\svv{B} $}
\psfrag{f1}[][][0.7]{$2\Nt$}
\psfrag{g1}[][][1]{$\svv{A}^{-1}$}
\includegraphics[width=2\columnwidth]{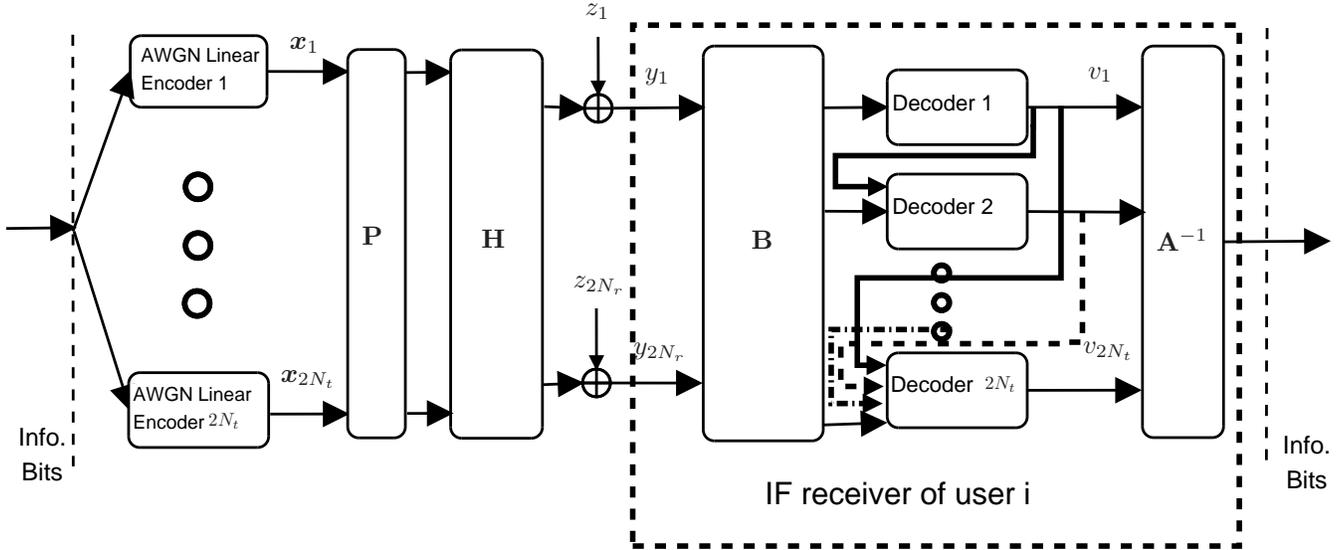}
\end{psfrags}
\end{center}
\caption{Linear pre-processed-IF / Linear pre-processed IF-SIC scheme. Feedback in the receiver is active only when SIC is used.}
\label{fig:integerForcing}
\end{figure*}

In \cite{IntegerForcing}, a receiver architecture scheme named ``integer forcing'' was proposed which we briefly recall. For our purposes, it will suffice to only state the achievable rates of IF and a high-level operational description of its elements. The reader is referred to  \cite{IntegerForcing} for the derivation, details and proofs, and further to  \cite{nazer2011compute,PracticalCodeDesign_OrdentlichErez:2011} and references therein for implementation considerations.

We follow the derivation of \cite{IntegerForcing} and describe integer forcing over the reals. Channel model~(\ref{eq:channel_model}) can be expressed via its real-valued representation as
\begin{align}
\underbrace{\begin{bmatrix}{\rm Re}({\bf y}_c) \\ {\rm Im}({\bf y}_c)\end{bmatrix}}_{{\bf y}}=\underbrace{\begin{bmatrix}{\rm Re}({\svv{H}_c}) && -{\rm Im}({\svv{H}_c}) \\
 {\rm Im}({\svv{H}_c}) && {\rm Re}({\svv{H}_c})\end{bmatrix}}_{\svv{H}}\underbrace{\begin{bmatrix}{\rm Re}({\bf x}_c) \\ {\rm Im}({\bf x}_c)\end{bmatrix}}_{{\bf x}}+\underbrace{\begin{bmatrix}{\rm Re}({\bf z}_c) \\ {\rm Im}({\bf z}_c)\end{bmatrix}}_{{\bf z}}.
\label{eq:realRepOfH}
\end{align}
This real-valued representation is used in the sequel to derive performance bounds for the complex channel $\svv{H}_c$. Note that the dimensions of $\svv{H}$ are $2 N_r\times2\Nt$.

It is assumed that information bits are fed into $2\Nt$ encoders, each of which uses the same \emph{linear} code that is designed for an AWGN channel.\footnote{The effect of the chosen code on the overall performance of IF is discussed in \cite{OrdentlichErez:IFUniversallyAchievesCapacityUpToGap:2013}.} The latter produces $2\Nt$ channel inputs (for example, ${x}_{m}$ for the $m$'th antenna).\footnote{For simplicity of notation the time index is suppressed as the block length plays
no role in our description. Of course, to approach capacity, one needs to use a long block.} At the receiver, a \emph{linear equalization} matrix $\svv{B}\in\mathbb{R}^{2\Nt\times 2N_r}$ is applied. It is easiest to understand IF  by first describing its zero-forcing variant. In this case,   $\svv{B} $ is designed so that the resulting equivalent channel $\svv{A}=\svv{B} \svv{H}$ is such that $\svv{A}\in\mathbb{Z}^{2\Nt\times2\Nt}$ is a full-rank integer matrix. In a practical implementation, it may be necessary for the matrix to be full-rank over a finite field $\mathbb{Z}_p$ (where $p$ is prime) over which the code is defined. Nonetheless, by taking $p$ large enough, it suffices for $\svv{A}$ to be invertible over the reals (see Lemma 2 in Appendix~A of \cite{ApproximateSumCapacity_OrdentlichErezNazer:2016}; see also \cite{PracticalCodeDesign_OrdentlichErez:2011}). This ensures that the output of the channel (without noise) after applying a modulo operation is a valid codeword.

Each of the equalized streams is next passed to a standard (up to the additional element of a modulo operation) AWGN decoder which tries to decode a linear combination of codewords, whose coefficients correspond
to a row of $\svv{A}$.
Finally, after the noise is removed, the original messages are recovered by applying the inverse of $\svv{A}$.
Thus, for IF equalization to be successful, decoding over all $2\Nt$ subchannels should be successful and the worst subchannel constitutes a bottleneck.
The operation of the receiver is depicted in Figure~\ref{fig:integerForcing} (where at this stage the linear pre-processing matrix can be considered as the identity matrix, i.e., $\svv{P}=\svv{I}$).

When using minimum mean square  error (MMSE) equalization, rather than zero-forcing, the linear equalizer takes the form
\begin{align}
\label{eq:Bint}
 \svv{B} =\svv{A}\svv{H}^T\left(\svv{I}+\svv{H}\svv{H}^T\right)^{-1},
\end{align}
and the input to the $m$'th decoder is
\begin{align}
 {\boldsymbol{y}}_{\rm eff,m}=\boldsymbol{v}_{m}+\boldsymbol{z}_{\rm eff,m}
\end{align}
where
\begin{align}
  \boldsymbol{z}_{\rm eff,m}=(\boldsymbol{b}_{m}^T\svv{H}-\boldsymbol{a}_{m}^T)\boldsymbol{x}+\boldsymbol{b}_{m}^T\boldsymbol{z}.
\end{align}
Here, ${\bf a}_{m}^T$ and ${\bf b}_{m}^T$ are the m'th row of $\svv{A}$ and $\svv{B}$ respectively. We can define the effective SNR at the $m$'th subchannel as
\begin{align}
 \SNR_{\rm eff}({\bf a}_{m})=\left(\boldsymbol{a}_{m}^T(\svv{I}+\svv{H}^T\svv{H})^{-1}\boldsymbol{a}_{m}\right)^{-1},
 \label{eq:SNReffM}
\end{align}
and the effective rate that can be achieved at the $m$'th subchannel as
\begin{align}
R_{\rm IF}(\svv{H};{\bf a}_{m})&=\frac{1}{2}\log\left(\SNR_{\rm eff}({\bf a}_{m})\right) \nonumber \\
&=-\frac{1}{2}\log\left(\boldsymbol{a}_{m}^T(\svv{I}+\svv{H}^T\svv{H})^{-1}\boldsymbol{a}_{m}\right).
\label{eq:Ram}
\end{align}
Note that the rate expression (\ref{eq:Ram}) is negative when $\SNR_{\rm eff}({\bf a}_{m}) < 1$. Hence, the achievable rate should be understood as the maximum between (\ref{eq:Ram}) and zero.

By Theorem~3 in~\cite{IntegerForcing}, transmission with IF equalization can achieve any rate satisfying $R<R_{\rm IF}(\svv{H})$ where

\begin{align}
 &R_{\rm IF}(\svv{H})
 =\max_{\substack{\svv{A}\in\mathbb{Z}^{2\Nt\times2\Nt}\\ \det{\svv{A}}\neq 0}}\min_{m=1,...,2\Nt}2\Nt \cdot R_{\rm IF}(\svv{H};{\bf a}_{m})  \nonumber \\
 &=\max_{\substack{\svv{A}\in\mathbb{Z}^{2\Nt\times2\Nt}\\ \det{\svv{A}}\neq 0}}\min_{m=1,...,2\Nt}2\Nt\frac{1}{2}\log\left(\SNR_{\rm eff}({\bf a}_{m})\right)\nonumber \\
  &=\max_{\substack{\svv{A}\in\mathbb{Z}^{2\Nt\times2\Nt}\\ \det{\svv{A}}\neq 0}}\min_{m=1,...,2\Nt}\Nt\log\left(\frac{1}{\boldsymbol{a}_{m}^T(\svv{I}+\svv{H}^T\svv{H})^{-1}\boldsymbol{a}_{m}}\right)\nonumber \\
 &=\Nt\log\left(\min_{\substack{\svv{A}\in\mathbb{Z}^{2\Nt\times2\Nt}\\ \det{\svv{A}}\neq 0}}\max_{m=1,...,2\Nt}\left(\boldsymbol{a}_{m}^T(\svv{I}+\svv{H}^T\svv{H})^{-1}\boldsymbol{a}_{m}\right)\right)^{-1} \nonumber \\
 &=-\Nt\log\left(\min_{\substack{\svv{A}\in\mathbb{Z}^{2\Nt\times2\Nt}\\ \det{\svv{A}}\neq 0}}\max_{m=1,...,2\Nt}\left(\boldsymbol{a}_{m}^T(\svv{I}+\svv{H}^T\svv{H})^{-1}\boldsymbol{a}_{m}\right)\right).
 \label{eq:IF_overTheReals}
\end{align}

The achievable rate of IF may also be described via the successive minima of a lattice associated with the channel matrix as we now recall. Any channel can be described via its SVD
\begin{align}
\svv{H}=\svv{U}\svv{\Sigma}\svv{V}^T.
\label{eq:svd}
\end{align}
Using (\ref{eq:svd}), the following decomposition is readily obtained
\begin{align}
(\svv{I}+\svv{H}^T\svv{H})^{-1}=\svv{V}\svv{D}^{-1}\svv{V}^T,
\label{eq:svd_I_HHt_ind}
\end{align}
where
$\svv{D}=\svv{I}+{\Sigma}^T{\Sigma}$.
It follows that \eqref{eq:Ram} may be rewritten as (Theorem 4 in \cite{IntegerForcing})
\begin{align}
R_{\rm IF}(\svv{H};{\bf a}_{m})&=-\frac{1}{2}\log\left(\|\svv{D}^{-1/2}\svv{V}^T{\bf a}_{m}\|^2\right) \nonumber \\
&\triangleq R_{m,\rm IF}(\svv{D},\svv{V}),
\label{eq:comprate_m}
\end{align}
where in the last equation the dependence on the choice of $\svv{A}$ is left implicit.
Let $\Lambda$ be the lattice spanned by $\svv{G}=\svv{D}^{-1/2}\svv{V}^T$and recall the definition of successive minima.
\begin{definition}
Let $\Lambda(\svv{G})$ be a lattice spanned by the full-rank matrix $\svv{G}\in\mathbb{R}^{K\times K}$. For $k = 1, ... , K$, we define the $k$'th successive minimum as
\begin{align}
\lambda_k(\svv{G})\triangleq\inf\left\{r:{\rm dim}\left({\rm span}\left(\Lambda(\svv{G})\cap\mathcal{B}_{K}(r)\right)\right)\geq k\right\}
\end{align}
where $\mathcal{B}_{K}(r)=\left\{{\bf x}\in\mathbb{R}^K:\|{\bf x}\|\leq r\right\}$ is the closed ball of radius $r$ around ${\bf 0}$. In words, the $k$-th successive minimum of a lattice is the minimal radius of a ball centered around ${\bf 0}$ that contains $k$ linearly independent lattice points.
\end{definition}
Thus, the maximal rate achievable with integer-forcing equalization (\ref{eq:IF_overTheReals}) may be
written as
\begin{align}
R_{\rm IF}(\svv{H}) &= R_{\rm IF}(\svv{D},\svv{V}) \nonumber  \\
&= -\Nt\log\left(\lambda_{2\Nt}^2(\Lambda)\right)  \nonumber \\
&=\Nt\log\left(\frac{1}{\lambda_{2\Nt}^2(\Lambda)}\right).
\label{eq:R_IF_SuccMin}
\end{align}

\subsection{Integer-Forcing Equalization With Successive Interference Cancellation}
\label{sec:IFequalizationSubSecWithSIC}

We also consider a version of IF equalization  incorporating successive interference cancellation. We will refer to it as IF-SIC.\footnote{We note that IF-SIC may in general allow using different rates per stream as stated in Theorem 5 in \cite{PracticalCodeDesign_OrdentlichErez:2011}. We  nevertheless assume throughout that
all streams are encoded via an identical linear code and hence have the same rate.}
We state only the achievable rates of IF-SIC and an operational description of its elements.
The reader is referred to \cite{OrdentlichErezNazer1:2013} for the derivation, details and proofs.

For a given choice of integer matrix $\svv{A}$, let $\svv{L}$ be defined by the following Cholesky decomposition
\begin{align}
    \label{eq:Kzz}
   \svv{A}\left(\svv{I}+\svv{H}^T\svv{H}\right)^{-1}\svv{A}^T&=\svv{A}\svv{V}\svv{D}^{-1}\svv{V}^T\svv{A}^T \nonumber\\
   &=\svv{L}\svv{L}^T.
\end{align}
Denoting by $\ell_{m,m}$ the diagonal entries of $\svv{L}$, IF-SIC can achieve (see \cite{OrdentlichErezNazer1:2013}) any rate satisfying $R<R_{\rm IF-SIC}(\svv{H})$ where
\begin{align}
 R_{\rm IF-SIC}(\svv{H})&=R_{\rm IF-SIC}(\svv{D},\svv{V}) \nonumber \\
 &=
 2\Nt \cdot \frac{1}{2}\max_{\svv{A}}\min_{m=1,...,2\Nt}\log\left(\frac{1}{\ell_{m,m}^2}\right),
 \label{eq:IF-SIC-rate}
\end{align}
and the maximization is over all full-rank $2\Nt\times 2\Nt$ integer  matrices.\footnote{We note
that since we choose to work with equal-rate streams, the constraints on
the achievable rate tuples of IF with SIC, as stated in
Theorem~2 of \cite{OrdentlichErezNazer1:2013}, play no role in the present work.}

We describe the operation of the IF-SIC receiver, adopting the nomenclature of \cite{OrdentlichErezNazer1:2013}.
We note that we describe the MMSE-GDFE version
of IF-SIC, as given in Appendix~A of \cite{OrdentlichErezNazer1:2013}, rather than its noise-prediction variant.
First, calculate:
\begin{enumerate}
\item The optimal integer matrix $\svv{A}$, i.e., the matrix  maximizing (\ref{eq:IF-SIC-rate}).
\item The covariance matrix (\ref{eq:Kzz}) of the effective noise (see (\ref{eq:SNReffM})) that arises when using the equalization matrix $\svv{B}$ as given in (\ref{eq:Bint}).
\item The optimal SIC matrix $\svv{S}$ as:
    \begin{eqnarray}
    \svv{S}={\rm diag}(\ell_{11},...,\ell_{MM})\cdot\svv{L}^{-1}.
    \end{eqnarray}
\item The optimal combined linear front-end processing matrix:
    \begin{eqnarray}
    \widetilde{\svv{B }}
    &=& \svv{S} \svv{B } \nonumber \\ &=&\svv{S}\svv{A}\svv{H}^T\left(\svv{I}+\svv{H}\svv{H}^T\right)^{-1}.
    \end{eqnarray}
\end{enumerate}
The operation of the receiver is depicted in Figure~\ref{fig:integerForcing} where the feedback depicted in the receiver is now active, and where now $\svv{B} $ is to be understood as $\widetilde{\svv{B}}$.
Note that this change of linear post-processing is essential to guarantee that the resulting noise variance is minimized.
The outputs of decoders $1,\ldots,m-1$ are multiplied by $S_{m,1},\ldots,S_{m,m-1}$, respectively, and are then subtracted from the input to decoder $m$, thereby performing SIC.
\subsection{Linear Pre-Processed Integer Forcing}
\label{sec:precIFequalizationWithSic}

\subsubsection{Motivating example: Performance comparison of linear MMSE and IF receivers}
As a motivating example, following Remark~\ref{remark1}, we compare the performance of linear MMSE and IF equalizers over
a specific ensemble of channels defined over $\mathbb{H}(C=8,N_t=2)$. Specifically, we consider a ``normalized" $2\times2$ Rayleigh fading ensemble, where the capacity is fixed to  $C=8$ bits. The ensemble is generated by drawing a $2 \times 2$ channel matrix with i.i.d. circularly symmetric complex Gaussian entries and then scaling the matrix (multiplying it by a value that we find by numerical search) such that the mutual information equals $8$ bits.\footnote{Note that in this ensemble, the probability of channels corresponding to $N_r \neq 2$ is zero.}
\begin{figure}
\begin{center}
\includegraphics[width=\columnwidth]{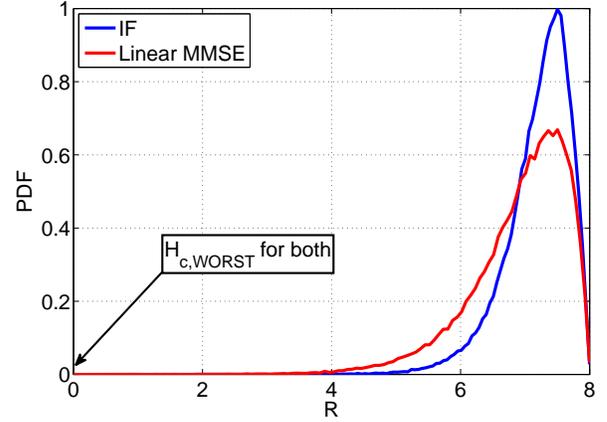}
\end{center}
\caption{Approximate probability density functions (based on Monte Carlo simulation)
of the  rates achievable with the linear MMSE receivers over a Rayleigh $2\times2$ MIMO channel normalized to $C=8$ bits.}
\label{fig:MotExample_0}
\end{figure}

Figure~\ref{fig:MotExample_0} depicts the probability density function of the rate achieved for this ensemble when using  linear MMSE and IF receivers. Since linear MMSE equalization is a special case of IF (setting $\svv{A}=\svv{I}$), as expected IF displays improved performance.

The real strength of IF lies however in the behavior of the ``tail".
For conventional linear equalizers, bad channels correspond to ill-conditioned matrices. An extreme case  is  the following channel  
\begin{align}
    \svv{H}_{\rm c,WORST}=\sqrt{2^8-1}\begin{bmatrix}1 & 0 \\ 0 & 0 \end{bmatrix}.
    \label{eq:bad_channel}
\end{align}
In this case, the data stream sent from the second antenna is completely lost when transmitted over the channel. Clearly, in this example, no receiver (including maximum likelihood) will be able to recover the lost data stream and thus the achievable rate of both linear and IF equalization is also zero.

Consider now the channel $\svv{H}_{\rm c,WORST}\cdot\svv{P}_c$ where $\svv{P}_c$ is a unitary matrix. As the singular values
remain unchanged, it is clear that the channel remains ill-conditioned and hence a linear receiver (not allowing for a modulo operation) will still
suffer from poor performance. On the other hand, the IF receiver performs well even over ill-conditioned MIMO channels, and in
fact, the performance of the IF receiver for the channel \eqref{eq:bad_channel} is good for ``most" pre-processing matrices as illustrated next.

\subsubsection{Linear pre-processing ensemble and resulting performance}
\label{sec:lin-pre}
The transmission scheme we analyze consists of applying a unitary  pre-processing matrix at the transmitter and IF equalization (either with or without SIC) at the receiver, as depicted in Figure~\ref{fig:integerForcing}. Applying
linear pre-processing  may be viewed as generating a ``virtual'' channel $\widetilde{\svv{H}}_c=\svv{H}_c\svv{P}_c$ over which transmission takes place. We restrict ourselves to unitary linear pre-processing matrices in order to keep the transmission power unchanged.

Throughout this paper, we assume that the linear pre-processing matrix $\svv{P}_c$ is drawn from what is referred to as the ``circular unitary ensemble''
(CUE). The ensemble is defined by the unique distribution on unitary matrices that is invariant under left and right unitary transformations (Theorem 8.3 in \cite{metha1967random}). In other words, the ensemble amounts to inducing the Haar measure on the unitary group of degree $N_t$.\footnote{An explanation on how
to generate matrices belonging to the CUE can be found in, e.g.,  \cite{mezzadri2006generate:2006}.}

\begin{remark}
While in general the  transformation ${\svv{P}}_c$ implies joint processing at the encoders, we note that in some natural statistical scenarios, including that of  an i.i.d. Rayleigh fading environment, the random transformation is actually performed by nature.\footnote{This follows since the left and right singular vector matrices of the an i.i.d. Gaussian matrix $\svv{H}_c$ are equal to the eigenvector matrices of the  Wishart ensembles $\svv{H}_c\svv{H}_c^{H}$ and $\svv{H}^H_c\svv{H}_c$, respectively. The latter are known to be CUE (Haar) distributed. See, e.g., Chapter~4.6 in \cite{edelman2005random}.} In such settings, our analysis holds even when the transmitters are distributed as in a multiple-access scenario.
\end{remark}



Figure~\ref{fig:MotExample} compares the achievable rates of the linear MMSE and IF receivers over the singular channel \eqref{eq:bad_channel}, when applying random CUE pre-processing. As can be seen, the
achievable rate of IF is high
for most pre-processing matrices, achieving a large fraction of $C$
with high probability.




\begin{figure}[h]
\begin{center}
\includegraphics[width=\columnwidth]{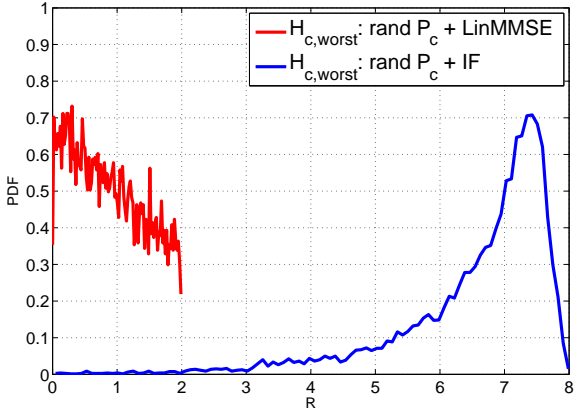}
\end{center}
\caption{Approximate probability density functions (based on Monte Carlo simulation) of the rates achievable with the linear MMSE and IF receivers over the channel \eqref{eq:bad_channel}, when applying a random linear pre-processing matrix drawn from the CUE.}
\label{fig:MotExample}
\end{figure}

\subsubsection{Properties of CUE pre-processing}
The SVD of the effective channel resulting from pre-processing is given by
\begin{align}
\svv{H}_c{\svv{P}}_c=\svv{U}_c\svv{\Sigma}_c{\svv{V}}_c^H{\svv{P}}_c.
\end{align}
Since ${\svv{V}}_c^H{\svv{P}}_c$ is equal in distribution to ${\svv{P}}_c$, for the sake of computing outage probabilities, we may simply assume that ${\svv{V}}_c^H$ (and also ${\svv{V}}_c$) is drawn from the CUE.



We note that the eigenvalue decomposition of the equivalent real channel can be written as
\begin{align}
(\svv{I}+\svv{H}^T\svv{H})^{-1}=\svv{V}\svv{D}^{-1}\svv{V}^T,
\end{align}
where
\begin{align}
{\svv{V}}=\begin{bmatrix}{\rm Re}({\svv{V}_c}) && -{\rm Im}({\svv{V}_c}) \\ {\rm Im}({\svv{V}_c}) && {\rm Re}({\svv{V}_c})\end{bmatrix}.
\end{align}
and
\begin{align}
\svv{D}=\begin{bmatrix} \svv{D}_c && \svv{0} \\ \svv{0} && \svv{D}_c \end{bmatrix}.
\label{eq:Dreal}
\end{align}
Further, the rates of IF, with or without SIC, for such a channel come in pairs.

Denoting the gap-to-capacity by $\Delta C$, we may therefore rewrite the worst-case IF outage probability as defined in (\ref{eq:P_WC_OUT}) as
\begin{align}
P^{\rm WC}_{\rm out,IF}\left(C,C-\Delta C\right)=\sup_{\svv{D}\in\mathbb{D}(C;2N_t)}\prob\left(R_{\rm IF}(\svv{D},\svv{V})<C-\Delta C\right)
\label{eq:rewirte6withD}
\end{align}
where we define $\mathbb{D}(C;2N_t)$ as the set of all $2N_t \times 2N_t$
diagonal matrices $\svv{D}$, with diagonal elements appearing in pairs, such that $\det\left(\svv{D}\right)=2^C$.

Another property we use in the sequel is the following. Denote by $d_{c,i}$ the diagonal entries of $\svv{D}_c$. Then
\begin{align}
2^C&=2^{\log\det\left(\svv{I}_{\Nt \times \Nt}+\svv{H}_c^H\svv{H}_c\right)} \nonumber\\
&=\det\left(\svv{V}_c\svv{D}_c\svv{V}_c^H\right) \nonumber \\
&=\prod_{i=1}^{\Nt}d_{c,i}.
\end{align}
Denoting by $d_i$ the diagonal entries of $\svv{D}$, we similarly have
\begin{align}
2^C&=2^{\frac{1}{2}\log\det\left(\svv{I}_{\Nt \times \Nt}+\svv{H}^T\svv{H}\right)} \nonumber \\
&=\sqrt{\det\left(\svv{V}\cdot\svv{D}\cdot\svv{V}^T\right)} \nonumber \\
&=\prod_{i=1}^{2\Nt}\sqrt{d_{i}}.
\end{align}
From (\ref{eq:Dreal}) we observe that since the singular values of the real channel come in pairs, we have
\begin{align}
\prod_{i=1}^{2\Nt}\sqrt{d_{i}}=\prod_{i=1}^{\Nt}d_{c,i}=2^C.
\end{align}
We denote $d_{\min}=\displaystyle\min_{i}d_i$ and  $d_{\max}=\displaystyle\max_{i}d_i$.

The following lemma will prove useful in characterizing the performance of CUE pre-processed IF. It relates the outage probability of CUE pre-processed IF to that arising when the pre-processing is performed using the circular real
ensemble (CRE).\footnote{The CRE is defined analogously to the CUE for the case of real orthonormal
matrices. That is, the ensemble is defined by the unique distribution on orthonormal matrices that is invariant under left and right orthonormal transformations.   
}
\begin{lemma}
Let $\svv{O}$ be a real $2\Nt \times 2\Nt$ matrix drawn from the CRE. Further, let $\mathbf{a}$ be a $2\Nt \times 1$ vector of integers. When
applying a random complex linear pre-processing matrix $\svv{V}_c$ that is drawn from the CUE (inducing a real-valued orthogonal pre-processing matrix $\svv{V}$), we have that $\left\|\svv{D}^{1/2}\svv{V}{\bf a}\right\|$ and $\left\|\svv{D}^{1/2}\svv{O}{\bf a}\right\|$
are equal in distribution.
\label{lem:lem1}
\end{lemma}
\begin{proof}
See Appendix~\ref{sec:proofOfLemma1}.
\end{proof}

%

%
\section{Bounds on the outage probability of CUE pre-processed Integer-Forcing}
\label{sec:BoundforNrxNtchannels}
\subsection{Derivation of Upper Bounds}
Define the dual lattice $\Lambda^*$ which is spanned by the matrix
\begin{align}
(\svv{G}^T)^{-1}=\svv{D}^{1/2}\svv{V}^T.
\label{eq:dualLattice}
\end{align}
Recall that the rate of IF is given by (\ref{eq:R_IF_SuccMin}). Now, the successive minima of $\Lambda$ and $\Lambda^*$ are related by (Theorem 2.4 in\cite{KorkinZolotarev:Lagarias1990})
\begin{align}
\lambda_1(\Lambda^*)^2\lambda_{2\Nt}(\Lambda)^2\leq\frac{2\Nt+3}{4}{\bar{\gamma}_{2\Nt}}^2,
\label{eq:Lagarias}
\end{align}
where $\bar{\gamma}_{2\Nt}$ is a ``monotonized" Hermite's constant as defined next.\footnote{In \cite{Banaszczyk}, Theorem 2.1, another bound for the relation between the successive minima of $\Lambda$ and $\Lambda^*$ is given. This bound is tighter for very large dimensions (it increases with $n^2$, whereas (\ref{eq:Lagarias}) increases with $n^3$). However, (\ref{eq:Lagarias}) has better constants and the cross between these expressions occurs only at $n=254$.} 
Hermite's constant is known only for dimensions $1-8$ and $24$. Since it has been never proved that $\gamma_{2\Nt}$ is monotonically increasing, we define
\begin{align}
\bar{\gamma}_{2\Nt}=\max\left\{\gamma_i:1\leq i \leq 2\Nt\right\}.
\end{align}
The tightest known upper bound for Hermite's constant, as derived in \cite{minimumValueOfQuadraticForms:Blichfeldt1929}, is
 \begin{align}
\gamma_{2\Nt}\leq\left(\frac{2}{\pi}\right)\Gamma\left(2+\Nt\right)^{1/\Nt}.
\label{eq:blich}
\end{align}
Since this is an increasing function of $N_t$, it follows that $\bar{\gamma}_{2\Nt}$ is smaller than the r.h.s. of (\ref{eq:blich}).\footnote{In the sequel we use the known values of Hermite's constant when possible, i.e. for $\Nt=2,3,4$. For other dimensions, we use this bound.}
Combining the latter with the exact values of
the Hermite constant for dimensions for which
it is known, we may lower bound the achievable rates of IF via the dual lattice as follows
\begin{align}
R_{\rm IF}(\svv{D},\svv{V})\geq\Nt\log\left(\frac{\lambda_{1}^2(\Lambda^*)}{\alNt}\right),
\end{align}
where
\begin{align}
\alNt=\begin{cases}
\frac{2\Nt+3}{4}\gamma_{2\Nt}^2,& \Nt=2,3,4,12 \\
\frac{2\Nt+3}{4}\left(\frac{2}{\pi}\Gamma\left(2+\Nt\right)^{1/\Nt}\right)^2,& {\rm otherwise}
\end{cases}.
\end{align}
Hence,
\begin{align}
&\prob\left(R_{\rm IF}(\svv{D},\svv{V}\right)<C-\Delta C) \nonumber \\
&\leq \prob\left(\Nt\log\left(\frac{\lambda_{1}^2(\Lambda^*)}{\alNt}\right)<C-\Delta C\right)  \nonumber \\
&=\prob\left(\lambda_{1}^2(\Lambda^*)<2^{\frac{C-\Delta C}{\Nt}}\alNt\right).
\label{eq:BallEq3}
\end{align}

The next lemma provides an upper bound on the outage probability as a function of the gap-to-capacity
$\Delta C$, the capacity $C$, and $d_{\min}$ (as well as the number of transmit antennas). Denote
\begin{align}
    \mathbb{A}(\genGam,d;2N_t) \triangleq \left\{{\bf a} \in \mathbb{Z}^{2N_t} :0<\|{\bf a}\|<\sqrt{\frac{\genGam}{d}}\right\}.
    \label{eq:A_beta_d}
\end{align}
\begin{lemma}
\label{lem:lem2}
For any complex Gaussian MIMO channel with $N_t$ transmit antennas and  with white-input mutual information $C$, i.e., $\svv{D}\in\mathbb{D}(C;2N_t)$, and for $\svv{V}_c$ drawn from the CUE (inducing a real-valued orthogonal pre-processing matrix $\svv{V}$), the outage probability of integer forcing is upper bounded by
\begin{align}
&\prob\left(R_{\rm IF}(\svv{D},\svv{V}\right)<C-\Delta C)  \nonumber \\
&\leq \sum_{{\bf a}\in\mathbb{A}(\genGam,d_{\min};2N_t)}\frac{2\Nt\left(2^{\frac{C-\Delta C}{\Nt}}\alNt\right)^{\Nt-1/2}}{\|{\bf a}\|^{2\Nt-1}2^C\frac{2}{\sqrt{d_{\min}}}},
\label{eq:47}
\end{align}
where
\begin{align}
\genGam=2^{\frac{C-\Delta C}{\Nt}}\alNt.
\label{eq:beta}
\end{align}
\label{eq:lemma2}
\end{lemma}
\begin{remark}
The summation in (\ref{eq:47}) may be greater than 1 for certain values of $C$ and $\Delta C$. Obviously, one may  take the minimum between this lemma and $1$ when bounding the outage probability.
\end{remark}
\begin{proof}
For a given $\genGam>0$, let us upper bound the probability
$\prob\left(\lambda_1^2(\Lambda^*)<{\genGam}\right)$
or equivalently
$\prob\left(\lambda_1(\Lambda^*)<\sqrt{\genGam}
\right)$.
Noting that the event $\left\{\lambda_1(\Lambda^*)<\sqrt{\genGam}\right\}$ is equivalent to the event
\begin{align}
\bigcup_{\mathbf{a}\in{\mathbb{Z}^{2\Nt}\setminus\{{\bf 0}\}}}\left\{||\svv{D}^{1/2}\svv{V}\mathbf{a}||<\sqrt{\genGam}\right\}
\label{eq:unionOfAllIntVectors}
\end{align}
and applying the union bound gives
\begin{align}
\prob\left(\lambda_1(\Lambda^*) < \sqrt{\genGam}\right)& \leq \sum_{{\bf a}\in\mathbb{Z}^{2\Nt} \setminus\{{\bf 0}\} }\prob\left(\|\svv{D}^{1/2}\svv{V}{\bf a}\|<\sqrt{\genGam}\right)
\nonumber
\\
&=\sum_{{\bf a}\in\mathbb{A}(\genGam,d_{\min};2N_t)}\prob\left(\|\svv{D}^{1/2}\svv{V}^T{\bf a}\|<\sqrt{\genGam}\right),
\label{eq:outageProbReal}
\end{align}
where the equality in (\ref{eq:outageProbReal}) follows since whenever \mbox{$\aNorm\cdot\sqrt{d_{\min}}\geq\sqrt{\genGam}$}, we have that $\prob\left(\|\svv{D}^{1/2}\svv{V}{\bf a}\|<\sqrt{\genGam}\right)=0$.

Let $\mathcal{S}$ denote the unit sphere of dimension $2N_t$, i.e.,
\begin{align}
\mathcal{S}=\left\{(x_1,x_2,\cdots,x_{2N_t}):x_1^2+x_2^2+\cdots+x_{2N_t}^2=1\right\}.
\end{align}
By Lemma~\ref{lem:lem1}
\begin{align}
\prob\left(\|\svv{D}^{1/2}\svv{V}{\bf a}\|<\sqrt{\genGam}\right)=\prob\left(\|\svv{D}^{1/2}\svv{O}\mathbf{a}\|<\sqrt{\genGam}\right).
\end{align}
Let ${\bf o}_{\|{\bf a}\|}\sim{\rm Unif}(\mathcal{S\cdot\|\mathbf{a}\|})$, and note that $\svv{O}\mathbf{a}$ is equal in distribution to ${\bf o}_{\|{\bf a}\|}$. It follows that
\begin{align}
\prob\left(\|\svv{D}^{1/2}\svv{V}{\bf a}\|<\sqrt{\genGam}\right)=\prob\left(\|\svv{D}^{1/2}{\bf o}_{\|{\bf a}\|}\|<\sqrt{\genGam}\right).
\label{eq:outProbOverTheReals}
\end{align}


Now the probability appearing on the r.h.s. of (\ref{eq:outProbOverTheReals}) has a simple geometric interpretation. 
Define an ellipsoid with axes $x_i=\sqrt{d_i}\cdot\|{\bf a}\|$ and denote its surface area by $L(x_1,x_2,...,x_{2\Nt})$. Then, the r.h.s. of (\ref{eq:outProbOverTheReals}) is the ratio of the part of the  surface area of an ellipsoid that lies inside a sphere of radius $\sqrt{\genGam}$ (denoted by ${\rm CAP_{\rm ell}}(x_1,x_2,...,x_{2\Nt})$) and the total surface area of the ellipsoid. This is illustrated in Figure~\ref{fig:fig1} for the case of two real dimensions.
We may rewrite (\ref{eq:outProbOverTheReals}) as
\begin{align}
\prob\left(\|\svv{D}^{1/2}{\bf o}_{\|{\bf a}\|}\|<\sqrt{\genGam}\right)&=\frac{|\svv{D}^{1/2}\mathcal{S}\cdot\|{\bf a}\|\cap\sqrt{\genGam}\mathcal{S}|}{\left|\svv{D}^{1/2}\|{\bf a}\|\mathcal{S}\right|}
\nonumber \\
&= \frac{{\rm CAP_{\rm ell}}(x_1,x_2,...,x_{2\Nt})}{L(x_1,x_2,...,x_{2\Nt})}.
\label{eq:ellipsCirecEq}
\end{align}

Neither the numerator nor the denominator of (\ref{eq:ellipsCirecEq}) has a closed-form expression. In order to upper bound
this ratio, we upper bound the numerator  $\rm CAP_{\rm ell}(x_1,x_2,...,x_{2\Nt})$ and lower bound the denominator (the surface area of the ellipsoid).
Using inequality (4.3) in \cite{carlson1966some} (see also inequality (57) and historical account in \cite{tee2004surface}), we have
\begin{align}
L(x_1,...,x_{2\Nt})&>{\rm Vol}(\mathcal{B}_{2\Nt}(1))\|{\bf a}\|^{2\Nt}\prod_{i=1}^{2\Nt}{\sqrt{d_i}}\sum_{i=1}^{2\Nt}{\frac{1}{\|{\bf a}\|\sqrt{d_i}}} \nonumber \\
&\geq{\rm Vol}(\mathcal{B}_{2\Nt}(1))\|{\bf a}\|^{2\Nt-1}2^C\frac{2}{\sqrt{d_{\min}}} \nonumber \\
& \triangleq \underline{L}(x_1,x_2,...,x_{2\Nt}),
\label{eq:underlineL}
\end{align}
where $\mathcal{B}_{2\Nt}(1)$ is a unit ball of dimension $2\Nt$, and
\begin{align}
{\rm Vol}(\mathcal{B}_{2\Nt}(1))=\frac{\pi^{\Nt}}{\Gamma(1+\Nt)}
\label{eq:Vol}
\end{align} is its volume.


As an upper bound for the numerator, we take the entire surface area of a sphere of radius $\sqrt{\genGam}$, which is given by 
\begin{align}
    A_{2\Nt}(\sqrt{\genGam})=2\Nt\frac{\pi^{\Nt}}{\Gamma(1+\Nt)}\sqrt{\genGam}^{2\Nt-1}.
\end{align}
We thus have
\begin{align}
    \rm CAP_{\rm ell}(x_1,x_2,...,x_{2\Nt})&\leq A_{2\Nt}(\sqrt{\genGam}) \nonumber \\
    &\triangleq{\overline{\rm CAP_{\rm ell}}(\sqrt{\genGam})}.
    \label{eq:upCapEll}
\end{align}
\begin{figure}
\begin{center}
\includegraphics[width=\columnwidth,trim={0 2.6cm 0 2.2cm},clip]{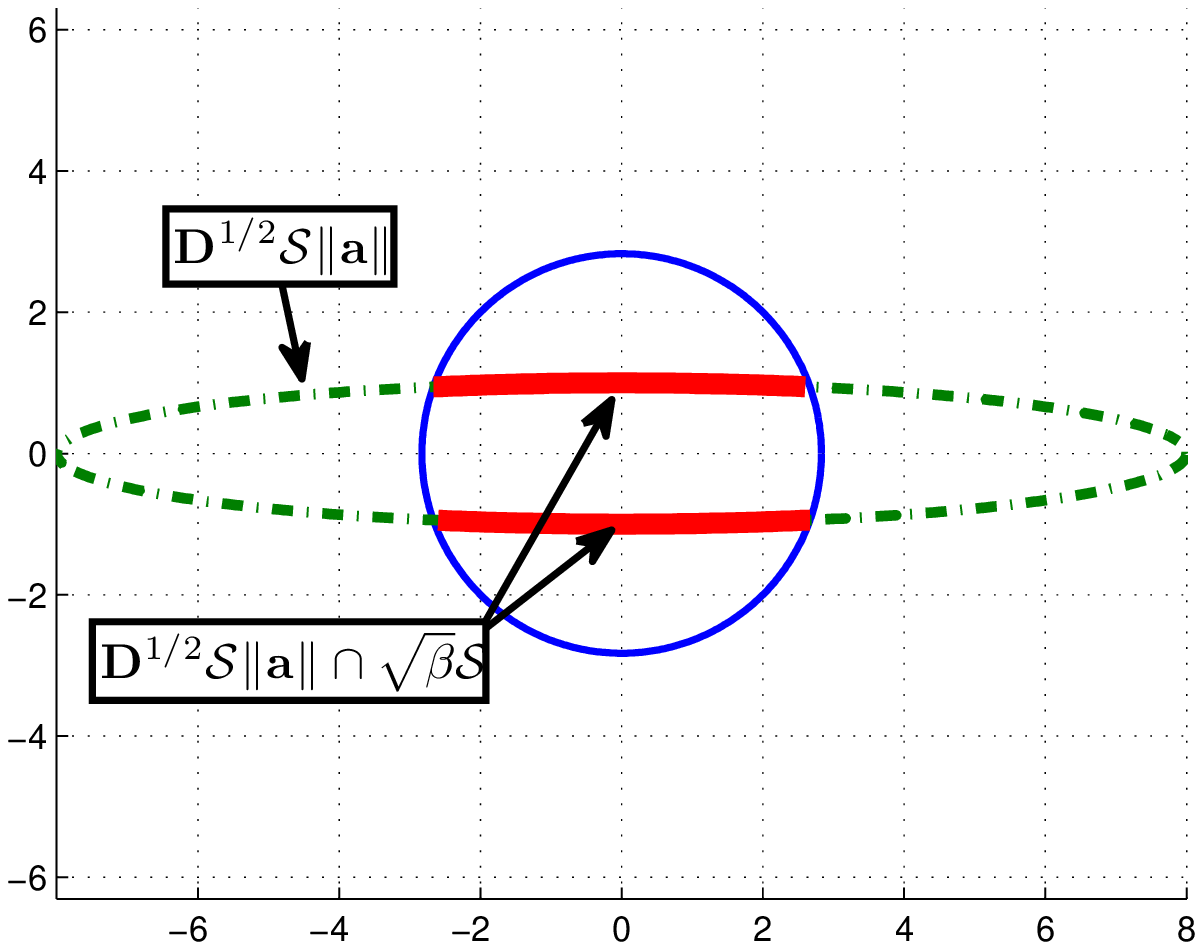}
\end{center}
\caption{Illustration of the geometric  objects appearing in (\ref{eq:ellipsCirecEq}).}
\label{fig:fig1}
\end{figure}
We may therefore bound (\ref{eq:ellipsCirecEq}) by
\begin{align}
\prob\left(\|\svv{D}^{1/2}{\bf o}_{\|{\bf a}\|}\|<\sqrt{\genGam}\right)\leq\frac{\overline{\rm CAP_{\rm ell}}(\sqrt{\genGam})}{\underline{L}(x_1,x_2,...,x_{2\Nt})}.
\label{eq:boudnWithUnderAndUp}
\end{align}
Substituting (\ref{eq:underlineL}), (\ref{eq:upCapEll}) into (\ref{eq:boudnWithUnderAndUp}) yields
\begin{align}
    &\sum_{{\bf a}\in\mathbb{A}(\genGam,d_{\min};2N_t)}\prob\left(\|\svv{D}^{1/2}{\bf o}_{\|{\bf a}\|}\|<\sqrt{\genGam}\right)  \nonumber \\
    &\leq \sum_{{\bf a}\in\mathbb{A}(\genGam,d_{\min};2N_t)}\frac{2\Nt\frac{\pi^{\Nt}}{\Gamma(1+\Nt)}\sqrt{\genGam}^{2\Nt-1}}{\frac{\pi^{\Nt}}{\Gamma(1+\Nt)}
        \|{\bf a}\|^{2\Nt-1}2^C\frac{2}{\sqrt{d_{\min}}}}  \nonumber \\
    & = \sum_{{\bf a}\in\mathbb{A}(\genGam,d_{\min};2N_t)}\frac{2\Nt\sqrt{\genGam}^{2\Nt-1}}{\|{\bf a}\|^{2\Nt-1}2^C\frac{2}{\sqrt{d_{\min}}}}.
    \label{eq:explicitSummation}
\end{align}
Substituting $\genGam=2^{\frac{C-\Delta C}{\Nt}}\alNt$, we finally arrive at
\begin{align}
 & \prob\left(R_{\rm IF}(\svv{D},\svv{V})<C-\Delta C\right) \nonumber \\
 & \leq \sum_{{\bf a}\in\mathbb{A}(\genGam,d_{\min};2N_t)}\prob\left(\|\svv{D}^{1/2}{\bf o}_{\|{\bf a}\|}\|<\sqrt{\genGam}\right)  \nonumber \\
 & \leq \sum_{{\bf a}\in\mathbb{A}(\genGam,d_{\min};2N_t)}\frac{2\Nt\left(2^{\frac{C-\Delta C}{\Nt}}\alNt\right)^{\Nt-1/2}}{\|{\bf a}\|^{2\Nt-1}2^C\frac{2}{\sqrt{d_{\min}}}}.
\end{align}
\end{proof}
The bound of Lemma~\ref{lem:lem2} is depicted in Figure~\ref{fig:lem1NoCountOpt}. Rather than plotting the outage probability, its complement is depicted, i.e., we plot the cumulative distribution function of the event that the rate is achieved by IF. For given $C$ and $\Delta C$, Lemma~\ref{lem:lem2} was numerically calculated over a grid of singular values. For each such vector of singular values, summation was performed over all ${\bf a}\in\mathbb{A}(\genGam,d_{\min};2N_t)$. The worst-case outage probability over all vectors of singular values from the grid is presented.

In addition, empirical (Monte Carlo) results are also plotted.
For each vector of singular values, a large number of random unitary matrices was drawn and the outage probability was calculated.
The integer matrix was derived using the LLL algorithm.\footnote{Advanced techniques are known (see, e.g.,\cite{sakzad:2013complex} and \cite{Fischer:2016}) that can be used to further improve the empirical results.}
The worst case outage probability over all tested (i.e., those belonging to the grid) singular values is presented.

 As a further reference, the figure also depicts the universal guaranteed gap-to-capacity derived in \cite{OrdentlichErez:IFUniversallyAchievesCapacityUpToGap:2013}, which for the case of $\Nt=2$ amounts to $\Delta C=15.24$ bits\cite{OrdentlichErez:IFUniversallyAchievesCapacityUpToGap:2013}.\footnote{This upper bound on the gap-to-capacity is guaranteed for a different coding scheme than that considered in this paper, where space-time pre-processing is employed. Nevertheless, it serves as a useful benchmark.}


While Lemma~\ref{lem:lem2} provides an explicit bound on the outage probability, in order to calculate it, one needs to go over all diagonal matrices in $\mathbb{D}(C;2N_t)$ and for each diagonal matrix, sum over all the relevant integer vectors in $\mathbb{A}(\genGam,d_{\min};2N_t)$.
Hence, the bound can be evaluated only for moderate values of capacity and for a small number of transmit antennas. The following theorem  provides (a looser) simple
closed-form bound. Furthermore, this bound does not depend on capacity but rather only on the number of transmit antennas and the gap-to-capacity.
\begin{figure}
\begin{center}
\includegraphics[width=\columnwidth]{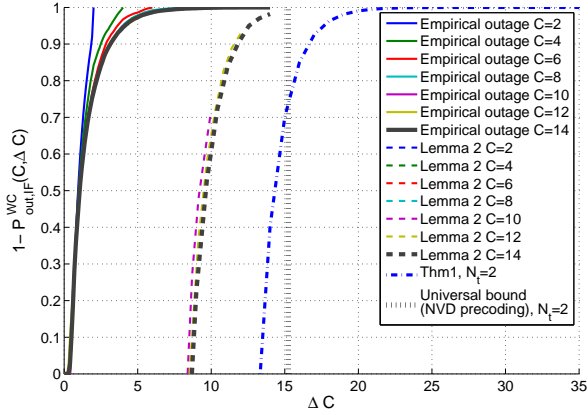}
\end{center}
\caption{Comparison of (worst-case) empirical results, Lemma~\ref{lem:lem2} and Theorem~\ref{thm:thm1} for two transmit antennas and for various values of WI mutual information.}
\label{fig:lem1NoCountOpt}
\end{figure}

\begin{theorem}
For any  complex Gaussian MIMO channel with $N_t$ transmit antennas and with WI mutual information $C$, and for $\svv{V}_c$ drawn from the CUE (inducing a real-valued linear pre-processing matrix $\svv{V}$), the outage probability of integer forcing is upper bounded by
\begin{align}
P^{\rm WC}_{\rm out,IF}\left(C,\Delta C\right)\leq c(\Nt)2^{-\Delta C},
\label{eq:thm1}
\end{align}
where
\begin{align}
    c(\Nt)=\left(2\Nt+\left(1+\sqrt{2\Nt}\right)^{2\Nt}\right)\Nt \alNt^{\Nt}\frac{\pi^{\Nt}}{\Gamma(\Nt+1)}
\end{align}
and
\begin{align}
\alNt=\frac{2\Nt+3}{4}\left(\frac{2}{\pi}\Gamma\left(2+\Nt\right)^{1/\Nt}\right)^2.
\end{align}
Thus, $c(\Nt)$ is a constant that depends only on $\Nt$.
\label{thm:thm1}
\end{theorem}
\begin{proof}
See Appendix~\ref{sec:proofOfTheorem1}.
\end{proof}

This bound is also depicted in Figure~\ref{fig:lem1NoCountOpt} for the case of two transmit antennas. For other values of $\Nt$, the bound is depicted in Figure~\ref{fig:fig2} (solid lines).\footnote{A slightly tightened version of Theorem~\ref{thm:thm1}, as described in Remark~\ref{rem:private} in Appendix~\ref{sec:proofOfTheorem1},  is used to generate Figures~\ref{fig:lem1NoCountOpt}~and~\ref{fig:fig2}.}
Recall again that, for $\Nt=2,3,4,$ we use the actual values of $\gamma_4=\sqrt{2}$, $\gamma_6=\left(\frac{64}{3}\right)^{1/6}$ and $\gamma_8=2$, rather than the bound of \cite{minimumValueOfQuadraticForms:Blichfeldt1929}. We note that when the number of transmit antennas increases, the gap between the theorem and the empirical results grows mainly due to the penalty incurred in (\ref{eq:Lagarias}) from using the dual lattice.

\begin{figure}
\centering
\includegraphics[width=\columnwidth]{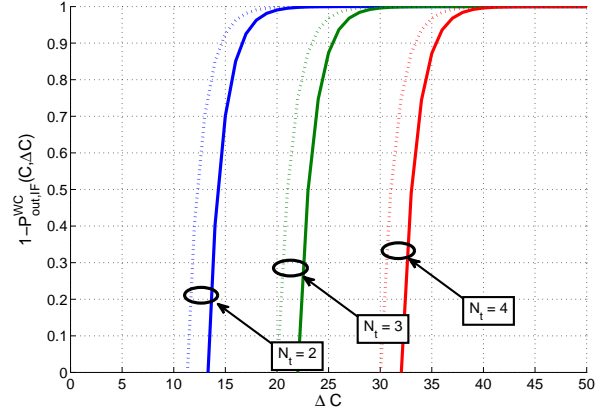}
\caption{The  performance guaranteed by Theorem~\ref{thm:thm1} for $\Nt=2,3,4$ transmit antennas. The dashed lines correspond to the improvement provided by Corollary 2.}
\label{fig:fig2}
\end{figure}

\subsection{Improved Upper Bounds}
A close inspection of Theorem~\ref{thm:thm1} reveals that there are two main sources for looseness in the bound that may be further tightened:
\begin{itemize}
\item
    \emph{Union bound} - While there is an inherent loss in the union bound, some terms in the summation (\ref{eq:outageProbReal}) may be completely dropped. See Corollary~\ref{col:col1} below.
\item
    \emph{Dual lattice} - Bounding via the dual lattice induces a loss reflected in (\ref{eq:Lagarias}). This may be circumvented for the case of two  transmit antennas, as accomplished (along with other improvements) in Theorem~\ref{thm:thm2} below.
\end{itemize}

We first tighten the union bound. As expressed in~(\ref{eq:unionOfAllIntVectors}), the event where the first of the successive minima is smaller than $\sqrt{\genGam}$ is equivalent to going over all integer vectors and checking whether any of them meet the norm condition. However, going over \emph{all} integer vectors is superfluous. In case that an integer vector $\mathbf{b}\in\mathbb{A}(\genGam,d_{\min};2N_t)$ is an integer multiple of another integer vector $\mathbf{a}\in\mathbb{A}(\genGam,d_{\min};2N_t)$, there is no need to count both of them. Rather, it suffices to  include in the union bound only the event corresponding to  $\mathbf{a}$. 

It follows that one may replace the set $\mathbb{A}(\genGam,d_{\min};2N_t)$ appearing in the summation in (\ref{eq:lemma2}) by a smaller set $\mathbb{B}(\genGam,d_{\min};2N_t)$ where
\begin{align}
&\mathbb{B}(\genGam,d;2N_t) \triangleq \nonumber \\
&\left\{{\bf a}\in \mathbb{Z}^{2N_t}:0<\|{\bf a}\|<\sqrt{\frac{\genGam}{d}} \:  {\rm and} \: \nexists 0<c<1 \: {\rm s.t.} \: c{\bf a}\in\mathbb{Z}^{2N_t} \right\}
\end{align}
as described by the next corollary.


\begin{corollary}
For any complex Gaussian MIMO channel with $N_t$ transmit  antennas and for $\svv{V}_c$ drawn from the  CUE (inducing a real-valued linear pre-processing matrix $\svv{V}$), the outage probability of integer forcing is upper bounded by
\begin{align}
 & \prob\left(R_{\rm IF}(\svv{D},\svv{V})<C-\Delta C\right) \nonumber \\
 & \leq \sum_{{\bf a}\in\mathbb{B}(\genGam,d_{\min} ;2N_t)}\frac{2\Nt\left(2^{\frac{C-\Delta C}{\Nt}}\alNt\right)^{\Nt-1/2}}{\|{\bf a}\|^{2\Nt-1}2^C\frac{2}{\sqrt{d_{\min}}}}.
\end{align}
where
$\genGam=2^{\frac{C-\Delta C}{\Nt}}\alNt$.
\label{col:col1}
\end{corollary}
%


A simpler restriction of the set $\mathbb{A}(\genGam,d;2N_t)$, short of reducing it to $\mathbb{B}(\genGam,d;2N_t)$, is obtained by noting that $\svv{D}$ and $\svv{V}$ are the real representations of complex matrices. Using the notations of (\ref{eq:realRepOfH}),  the integer vector ${\bf a}$ may be viewed as the real representation of the complex vector ${\bf a}_c$. Thus,
\begin{align}
\|\svv{D}_c^{1/2}\svv{V}_c{\bf a}_c\|=\|\svv{D}^{1/2}\svv{V}{\bf a}\|.
\end{align}
As multiplication of ${\bf a}_c$ by $\{-1,j,-j\}$ does not change the value of $\|\svv{D}_c^{1/2}\svv{V}_c{\bf a}_c\|$ (and equivalently, it does not change the value of $\|\svv{D}^{1/2}\svv{V}{\bf a}\|$), it suffices to include only one of these members of $\mathbb{A}(\genGam,d;2N_t)$ in the summation. Hence, a simple multiplicative improvement may be obtained.
\begin{corollary}
For any complex Gaussian MIMO channel with $N_t$ transmit antennas and for $\svv{V}_c$ drawn from the  CUE (inducing a real-valued linear pre-processing matrix $\svv{V}$), the outage probability of integer forcing is upper bounded by
\begin{align}
& \prob\left(R_{\rm IF}(\svv{D},\svv{V})<C-\Delta C\right) \nonumber  \\
& \leq \frac{1}{4}\sum_{{\bf a}\in\mathbb{A}(\genGam,d_{\min};2N_t)}\frac{2\Nt\genGam^{\Nt-1/2}}{\|{\bf a}\|^{2\Nt-1}2^C\frac{2}{\sqrt{d_{\min}}}},
\end{align}
where $\genGam=2^{\frac{C-\Delta C}{\Nt}}\alNt$.
\label{col:col2}
\end{corollary}

While the improvement of Corollary~\ref{col:col1} depends on $\genGam$ (and hence also on $C$), we may tighten Theorem~\ref{thm:thm1} by invoking Corollary~\ref{col:col2} as shown by the dashed lines in Figure~\ref{fig:fig2}. For a given value of $C$, we may combine the two corollaries. Figure~\ref{fig:2x2complexShowingAddedValueOfCol1} shows the different bounds on the outage probability for the case of a MIMO channel with two  antennas and with $C=14$, where both corollaries are utilized for tightening Lemma~\ref{lem:lem2}.



\begin{figure}
\begin{center}
\includegraphics[width=\columnwidth]{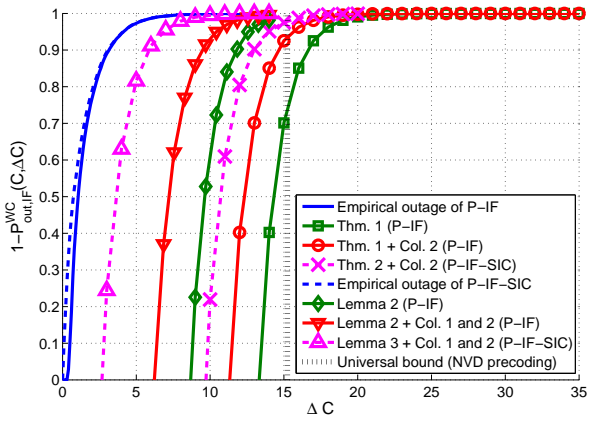}
\end{center}
\caption{Outage bounds for channels with $N_t=2$ transmit antennas and with WI mutual information $C=14$ bits.}
\label{fig:2x2complexShowingAddedValueOfCol1}
\end{figure}

As mentioned above, there is an additional significant loss due to using the dual lattice for deriving both Lemma~\ref{lem:lem2} and Theorem~\ref{thm:thm1}. For the case of $\Nt=2$, this loss may be circumvented by analyzing the performance of IF-SIC. When using IF-SIC, (\ref{eq:P_WC_OUT}) can be rewritten as
\begin{align}
& P_{\rm out,IF-SIC}^{\rm WC}\left(C,C-\Delta C\right) \nonumber \\
& = \sup_{\svv{D}\in\mathbb{D}(C;2N_t)}\prob\left(R_{\rm IF-SIC}(\svv{D},\svv{V})<C-\Delta C\right).
\label{eq:rewirte6withD_SIC}
\end{align}
The next lemma provides a bound on the outage probability of IF-SIC.

\begin{lemma}
\label{lem:lem3}
For any complex Gaussian MIMO channel with $N_t$ transmit antennas and with white-input mutual information $C>1$, i.e., $\svv{D}\in\mathbb{D}(C;2N_t)$, and for $\svv{V}_c$ drawn from the CUE (inducing a real-valued linear pre-processing matrix $\svv{V}$), the outage probability of integer forcing with successive interference cancellation is upper bounded by
\begin{align}
& \prob\left(R_{\rm IF-SIC}(\svv{D},\svv{V})<C-\Delta C\right) \nonumber \\
& \leq  \sum_{{\bf a}\in\mathbb{A}(\genGam,1/d_{\max};2N_t)}\frac{2\pi^2{2^{\nicefrac{-3}{4}(C+\Delta C)}}}{{\pi^2}\frac{\aNorm^3}{2^C}{\sqrt{d_{\max}}}},
\end{align}
where $\mathbb{A}(\genGam,1/d_{\max};2N_t)$ is defined in (\ref{eq:A_beta_d}), $\genGam=2^{\nicefrac{-1}{2}(C+\Delta C)}$ and for all $\Delta C>1$.

\end{lemma}
\begin{proof}
See Appendix~\ref{sec:BoundFor2x2channels}.
\end{proof}

\begin{remark}
Lemma~\ref{lem:lem3} can be further tightened using Corollary~\ref{col:col1}, i.e., by replacing $\mathbb{A}(\genGam,1/d_{\max};2N_t)$ with $\mathbb{B}(\genGam,1/d_{\max};2N_t)$.
\end{remark}

In a similar manner to the derivation of  Theorem~\ref{thm:thm1} using Lemma~\ref{lem:lem2}, for IF-SIC, Lemma~\ref{lem:lem3} leads to the following theorem.
\begin{theorem}
\label{thm:thm2}
For any complex Gaussian MIMO channel with $N_t$ transmit antennas
and with white-input mutual information $C>1$, i.e., $\svv{D}\in\mathbb{D}(C;2N_t)$, and for $\svv{V}_c$ drawn from the CUE (inducing a real-valued linear pre-processing matrix $\svv{V}$),  the outage probability of integer forcing with successive interference cancellation is upper bounded by
\begin{align}
P^{\rm WC}_{\rm out,IF-SIC}\left(C,\Delta C\right)\leq 85\pi^2 2^{-\Delta C},
\end{align}
for all $\Delta C>1$.
\end{theorem}
\begin{proof}
See Appendix~\ref{sec:BoundFor2x2channels}.
\end{proof}

Figure~\ref{fig:2x2complexShowingAddedValueOfCol1} depicts the improved bounds of Lemma~\ref{lem:lem3}
(incorporating the improvements provided by Corollary~\ref{col:col1} and Corollary~\ref{col:col2}) and Theorem~\ref{thm:thm2} for a system employing IF-SIC with $\Nt=2$.\footnote{A slightly tightened version of Theorem~\ref{thm:thm2}, as described in Remark~\ref{rem:rem7} in Appendix~\ref{sec:BoundFor2x2channels}, is used to generate Figure~\ref{fig:2x2complexShowingAddedValueOfCol1}.}


\subsection{Lower Bound via Maximum-Likelihood Decoding}
\label{sec:ml}
Beyond the upper bounds on performance derived thus far, it is natural to compare the worst-case performance attained by an IF receiver with that of an optimal maximum likelihood (ML) decoder for the same randomly linear pre-processed scheme but where each stream
is coded using an independent Gaussian codebooks. This provides a lower
bound on the worst-case outage probability of IF.

Consider a specific $N_r \times N_t$ matrix $\svv{H}_c$ and let $\svv{H}_S$ denote the submatrix of $\svv{H}_c\svv{P}_c$ formed by taking the columns with indices in $S\subseteq \{1, 2,..., \Nt \}$. For a joint ML decoder, the following is the maximal  rate achievable \cite{IntegerForcing} over the considered MIMO multiple-access channel:
\begin{align}
    R_{\rm JOINT}=\min_{S\subseteq\{1, 2,\ldots,\Nt\}}\frac{\Nt}{|S|}\log\det\left(\svv{I}_{N_r}+\svv{H}_S\svv{H}_S^H\right).
    \label{eq:R_ML_spaceOnly}
\end{align}
Note that since $\svv{H}_S$ depends on the random
linear pre-processing matrix $\svv{P}_c$, $ R_{\rm JOINT}$ is a random variable.
The lower bound is therefore obtained by taking the infimum
of \eqref{eq:R_ML_spaceOnly} over all $\svv{H}_c$ in $\mathbb{H}(C;N_t)$.


\begin{figure}
\begin{center}
\includegraphics[width=\columnwidth]{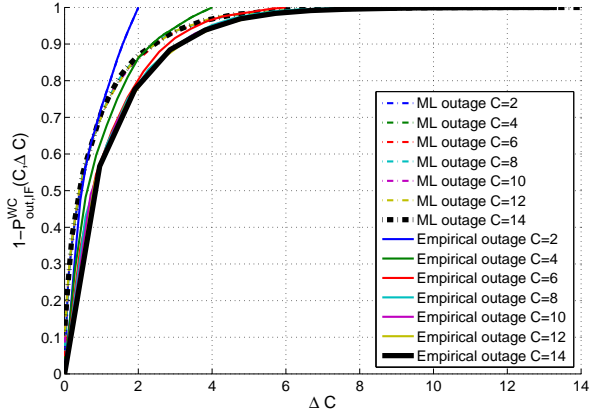}
\end{center}
\caption{Worst-case performance comparison between IF-SIC and joint ML decoding, for the case of $N_t=2$ transmit antennas.}
\label{fig:jointDecoding}
\end{figure}

Figure~\ref{fig:jointDecoding} provides a comparison between the worst-case empirical performance of IF-SIC and the (worst-case) empirical performance of the corresponding scheme with ML decoding, for the case of $N_t=2$ and an arbitrary number of receive antennas. In both cases, CUE pre-processing is applied.
Thus, performance depends only on the singular values of the channel and hence the outage probability curves are the supremum of the outage probability over a grid of (two) singular values.

As can be seen, the gap between IF and ML is quite small. This suggests that most of the loss with respect to the WI mutual information is due to the separate encoding of the data streams (i.e., MIMO MAC) rather than the suboptimailty of the IF receiver.

\section{Application : Universal Gap-To-Capacity for Multi-User Closed-Loop Multicast using P-IF}
\label{sec:BoundForClosedLoopMulticast}
Closed-loop MIMO multicast is a scenario where a transmitter equipped with $\Nt$ transmit antennas wishes to send the same message to $K$ users, where user $i$ is equipped with $N_{i}$ antennas.

Even though channel state information is available at both transmission ends, designing practical capacity-approaching schemes for closed-loop MIMO multicast with $K \geq 3$ users is challenging as detailed in \cite{khina:2015}. 
The outage bound derived above suggests that pre-processed IF may be an attractive practical closed-loop MIMO multicast scheme, allowing to obtain  a small gap-to-capacity with space-only pre-processing. Namely, we use the probabilistic method to establish the existence of  a pre-processing matrix guaranteeing
a desired gap-to-capacity.


We denote by $\svv{H}_{c,i}$ the $N_{i}\times \Nt$ channel matrix corresponding to the $i$th user and by $\mathcal{H}=\{\svv{H}_{c,i}\}_{i=1}^{K}$ the set of channels. The received signal at user $i$ is
\begin{eqnarray}
\boldsymbol{y}_{c,i}=\svv{H}_{c,i}\boldsymbol{x}_c+\boldsymbol{z}_{c,i}.
\label{eq:channel_modelClosed}
\end{eqnarray}
We assume that channel state information (CSI) is available at both transmission ends.

The multicast capacity is defined as the capacity of the compound channel~(\ref{eq:channel_modelClosed}). It is attained by a Gaussian input vector, where the mutual information is maximized over all covariance matrices  $\svv{Q}_c$ satisfying $ \Tr(\svv{Q}_c) \leq \Nt $:
\begin{eqnarray}
C(\mathcal{H})=\max_{\svv{Q}_c: \Tr(\svv{Q}_c) \leq \Nt}\min_{\svv{H}_c \in \mathcal{H}} \log\det(\svv{I}_{N_r \times N_r}+\svv{H}_c\svv{Q}_c\svv{H}^H_c).
\label{eq:multicastCapacity}
\end{eqnarray}
We assume without loss of generality that the input covariance matrix is the identity matrix. We may do so since the covariance shaping matrix $\svv{Q}_c^{1/2}$ may be absorbed into the channel by defining the effective channel $\hat{\svv{H}}_{c,i} = \svv{H}_{c,i}\svv{Q}_c^{1/2}$. Thus,
\begin{align}
C(\mathcal{H}) & =\min_i
\log\det(\svv{I}_{N_r \times N_r}+\hat{\svv{H}}_{c,i} \hat{\svv{H}}^H_{c,i}) \nonumber \\
& =\min_i
\log\det(\svv{I}_{N_t \times N_t}+\hat{\svv{H}}_{c,i}^H \hat{\svv{H}}_{c,i}).
\label{eq:closeLoopCapAsFuncOfWIofEachChannel}
\end{align}
In other words, after finding the optimal covariance matrix $\svv{Q}_c$, when it comes to the transmission scheme, it suffices to consider WI transmission over the effective channels $\hat{\svv{H}}_{c,i}$. With a slight abuse of notation, we use $\svv{H}_{c,i}$ to denote the effective channel, i.e., we drop the hat.
We note that for each user $i$, there exists an $\alpha_i\geq1$ such that
\begin{align}
{\svv{H}}_{c,i} = \alpha_i{\breve{\svv{H}}_{c,i}}.
\label{eq:usersWithEccSNR}
\end{align}
where
\begin{align}
\widebreve{\mathcal{H}}=\{\breve{\svv{H}}\}_{i=1}^{K}\in\mathbb{H}(C(\mathcal{H});\Nt),
\label{eq:breveInMathclH}
\end{align}
 i.e., $\{\breve{\svv{H}}\}_{i=1}^{K}$ is contained in the (continuum) set of channels, having the same capacity $C(\mathcal{H})$. Further, $\alpha_i$ can be interpreted as excess $\SNR$ that user $i$ enjoys, beyond the minimum it needs in the multicast setting.
 Since the achievable rate of IF is monotonically increasing in $\SNR$, 
 it follows that the achievable rates over the set of channels $\mathcal{H}$ can only be higher than over $\widebreve{\mathcal{H}}$, which we next lower bound.




Let us consider applying the random CUE pre-processed IF scheme to the compound channel set $\widebreve{\mathcal{H}}$.\footnote{We assume IF-SIC is used for $\Nt=2$ since it provides improved bounds.}
Define $A_i(R)$ as the event where the pre-processing matrix $\svv{P}_c$ is such that IF achieves a desired target $R$ for user $i$
\begin{align}
A_i(R)=\left\{{\svv{P}_c} : R_{\rm IF}(\svv{H}_{c,i}\cdot{\svv{P}}_c)\geq R \right\}.
\end{align}
We are interested in the probability of achieving the target rate for all users, i.e., $\prob\left(\cap{A_i(R)}\right)$. Note that
\begin{align}
\prob\left(\cap{A_i(R)}\right)=1-\prob\left(\overline{\cap{A_i(R)}}\right)=1-\prob\left(\cup{\overline{A_i(R)}}\right).
\end{align}
Applying the union bound, we get
\begin{align}
\prob\left(\cup{\overline{A_i(R)}}\right)\leq\sum\prob\left(\overline{A_i(R)}\right)
\end{align}
and hence
\begin{align}
\prob\left(\cap{A_i(R)}\right)\geq 1 - \sum\prob\left(\overline{A_i(R)}\right).
\end{align}
Define
\begin{align}
\widebreve{A_i}(R)=\left\{{\svv{P}_c} : R_{\rm IF}(\breve{\svv{H}}_{c,i}\cdot{\svv{P}}_c)\geq R \right\}.
\end{align}
Since $\prob\left(\widebreve{A_i}(R)\right)$ is the probability of achieving the target rate, whereas $P^{\rm WC}_{\rm out,IF}$ bounds the probability of the complement event, we have
\begin{align}
\prob\left(A_i(R)\right)\geq\prob\left(\widebreve{A_i}(R)\right)\geq1-P^{\rm WC}_{\rm out,IF}\left(C(\mathcal{H}),C-R\right),
\end{align}
or equivalently,
\begin{align}
\prob\left(\overline{A_i(R)}\right)=1-\prob\left(A_i(R)\right)\leq P^{\rm WC}_{\rm out,IF}\left(C(\mathcal{H}),C-R\right).
\end{align}
It follows that,
\begin{align}
\prob\left(\cap{A_i(R)}\right)\geq 1 - KP^{\rm WC}_{\rm out,IF}\left(C(\mathcal{H}),C-R\right).
\label{eq:capAi}
\end{align}

\begin{figure}
\includegraphics[width=\columnwidth]{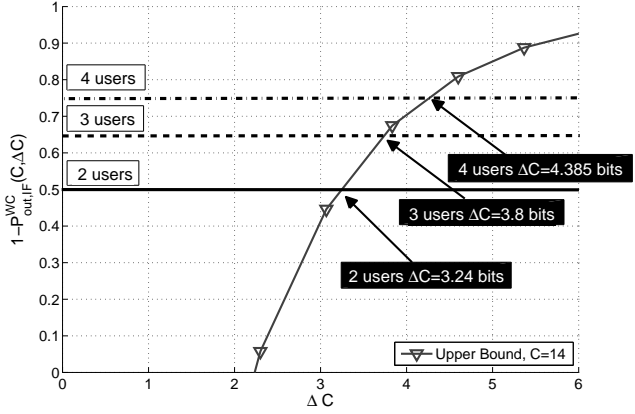}
\caption{Guaranteed achievable rates (using upper bound on outage probability stated in Corollary~\ref{col:col1} of Lemma~\ref{lem:lem3}) with pre-processed IF for closed-loop MIMO multicast transmission with two transmit antennas and with 2,3,4 users.}
\label{fig:empricalVsBoundFor2x2complex_closedLoop}
\end{figure}

This provides a means to obtain a guaranteed achievable transmission rate $R_{\rm WC-CL}(\mathcal{H})$ for closed-loop linear pre-processed IF. Namely, $R_{\rm WC-CL}(\mathcal{H})$ is the maximum rate for which
\begin{align}
P^{\rm WC}_{\rm out,IF}\left(C(\mathcal{H}),C-R\right)\leq\frac{1}{K}.
\label{eq:R_WC_CL}
\end{align}
Substituting~(\ref{eq:R_WC_CL}) in (\ref{eq:capAi}) we get that for any $R<R_{\rm WC-CL}(\mathcal{H})$
\begin{align}
\prob\left(\cap{A_i(R)}\right)> 1 - K\cdot \frac{1}{K}=0.
\end{align}
Thus, there must exist a linear pre-processing matrix $\svv{P}_c$ for which a target rate  $R<R_{\rm WC-CL}(\mathcal{H})$  is achievable (via P-IF transmission) for the compound channel~(\ref{eq:channel_modelClosed}).


Figure~\ref{fig:empricalVsBoundFor2x2complex_closedLoop} depicts the corresponding upper bounds on the gap-to-capacity for MIMO multicast with two transmit antennas, $K=2,3,4$ users, and where $C(\mathcal{H})=14$ bits. For calculating the upper bound on the guaranteed gap-to-capacity, we use the tightest bound on the outage probability we have developed for $\Nt=2$ which is Corollary~\ref{col:col1} of Lemma~\ref{lem:lem3}. We observe that
\begin{itemize}
\item For 2 users, a rate of 10.76 bits is guaranteed (gap of 3.24 bits to capacity).
\item For 3 users, a rate of 10.2 bits is guaranteed (gap of 3.8 bits to capacity).
\item For 4 users, a rate of 9.615 bits is guaranteed (gap of 4.385 to  capacity).
\end{itemize}

\section{Conclusion}
We obtained explicit universal bounds for the outage probability of a transmission scheme employing random unitary pre-processing  at the transmitter side and integer-forcing equalization at the receiver side. These bounds provide meaningful performance guarantees for
transmission over MIMO channels that depend only on
the channel's mutual information and number of transmit antennas. Nonetheless, simulations suggest that there is still a considerable gap between the obtained bounds and the true (worst-case) outage probability of the examined scheme, calling for further work.
\section{Acknowledgement}
The authors are deeply grateful to Or Ordentlich whose work laid the foundation for this paper and for many fruitful discussions.
\begin{appendices}
\section{Proof of Lemma~\ref{lem:lem1}}
\label{sec:proofOfLemma1}
We start by expressing $\left\|\svv{D}^{1/2}\svv{V}{\bf a}\right\|$ equivalently in complex notation. We note that ${\bf a}$ (which is a vector of $2\Nt$ real integers) can be viewed as the real representation of a complex vector ${\bf a}_c$ such that
 \begin{align}
{\bf a}=\begin{bmatrix}{\rm Re}({\bf a}_c) \\ {\rm Im}({\bf a}_c)\end{bmatrix}.
\end{align}
Obviously, $\|{\bf a}\|=\|{\bf a}_c\|$.

With this notation, and since $\svv{D}$ and $\svv{V}$ are the real representation of the complex matrices $\svv{D}_c$ and $\svv{V}_c$, it follows that
\begin{align}
\left\|\svv{D}^{1/2}\svv{V}{\bf a}\right\|=\left\|\svv{D}_c^{1/2}\svv{V}_c{\bf a}_c\right\|.
\end{align}
Now since $\svv{V}_c$ is drawn from the CUE, the distribution of $\left\|\svv{D}_c^{1/2}\svv{V}_c{\bf a}_c\right\|$ is equal to that of \mbox{$\left\|\svv{D}_c^{1/2}\svv{V}_c\begin{bmatrix} 1 && 0 && ... & 0\end{bmatrix}^T \|\mathbf{a}_c\|\right\|$}. Note also that
\begin{align}
&\left\|\svv{D}_c^{1/2}\svv{V}_c\begin{bmatrix} 1 && 0 && ... & 0\end{bmatrix}^T \|\mathbf{a}_c\|\right\|= \nonumber \\
&\left\|\svv{D}_c^{1/2}\svv{V}_c\begin{bmatrix} 1 && 0 && ... & 0\end{bmatrix}^T \right\|\|\mathbf{a}\|= \nonumber \\
&\left\|\svv{D}_c^{1/2}{\bf v}_{c,1}\right\|\|\mathbf{a}\|.
\label{eq:101}
\end{align}
where ${\bf v}_{c,1}$ is the first column of $\svv{V}_c$.

As described in \cite{NarulaTrottWornell:1999}, ${\bf v}_{c,1}$ is  uniformly distributed over the surface of the
complex unit sphere. Such a vector can be generated by taking a vector with zero-mean i.i.d. complex Gaussian components and scaling it
by its norm. The components of such a vector can be expressed as
\begin{align}
{\bf v}_{c,1}=\frac{G_{c,i}}{\sqrt{\sum_{i=1}^{\Nt}|G_{c,i}|^2}},
\end{align}
where $G_{c,i}$ are zero-mean i.i.d. complex circularly symmetric Gaussian random variables.

Similarly, a vector taken from a CRE matrix is uniformly distributed over the surface of the real unit sphere and it can be generated by taking a vector with zero-mean i.i.d. real Gaussian components and scaling it
by its norm. The components of such a vector can be expressed as
\begin{align}
{\bf o}_{r,1}=\frac{G_{r,i}}{\sqrt{\sum_{i=1}^{\Nt}G_{r,i}^2}},
\label{eq:narula}
\end{align}
where $G_{r,i}$ are zero-mean i.i.d. real Gaussian random variables.

We may rewrite  \eqref{eq:101} over the reals as
\begin{align}
\left\|\svv{D}^{1/2}\begin{bmatrix}{\rm Re}({\bf v}_{c,1}) \\ {\rm Im}({\bf v}_{c,1})\end{bmatrix}\right\|\|\mathbf{a}\|.
\label{eq:103}
\end{align}
Now, since the real and imaginary part of the complex Gaussian components are i.i.d. real Gaussian random variables, it follows that the resulting
$2\Nt\times1$ vector
\begin{align}
{\bf o}=\begin{bmatrix}{\rm Re}({\bf v}_{c,1}) \\ {\rm Im}({\bf v}_{c,1})\end{bmatrix}
\end{align}
is of the form of (\ref{eq:narula}). Hence, it is uniformly distributed over the surface of the ($2\Nt$-dimensional) real  unit  sphere and thus it can be interpreted as the first vector from a real matrix $\svv{O}$ drawn from CRE ensemble. Therefore
\begin{align}
\left\|\svv{D}^{1/2}{\bf o}\right\|\|\mathbf{a}\|
\label{eq:Do}
\end{align}
which equals (\ref{eq:101}) has  the same distribution as
\begin{align}
    \left\|\svv{D}^{1/2}\svv{O}\mathbf{a}\right\|.
\end{align}
It follows that $\left\|\svv{D}^{1/2}\svv{V}{\bf a}\right\|$ and $\left\|\svv{D}^{1/2}\svv{O}{\bf a}\right\|$ have the same distribution.
\section{Proof of Theorem~\ref{thm:thm1}}
\label{sec:proofOfTheorem1}
From Lemma~\ref{lem:lem2}, we have
\begin{align}
 & \prob(R_{\rm IF}\left(\svv{D},\svv{V}\right)<C-\Delta C) \nonumber \\
 & \leq \sum_{{\bf a}\in\mathbb{A}(\genGam,d_{\min};2N_t)}\frac{2\Nt\left(2^{\frac{C-\Delta C}{\Nt}}\alNt\right)^{\Nt-1/2}}{\|{\bf a}\|^{2\Nt-1}2^C\frac{2}{\sqrt{d_{\min}}}},
\end{align}
where $\mathbb{A}(\genGam,d_{\min};2N_t)$ and $\genGam$ are defined in (\ref{eq:A_beta_d}) and (\ref{eq:beta}). Reverting back to (\ref{eq:explicitSummation})
and noting that
$$\mathbb{A}(\genGam,d_{\min};2N_t) \subseteq \left\{\svv{a} \in \mathbb{Z}^{2N_t}: \|\svv{a}\| \leq \left\lfloor\sqrt{\frac{\genGam}{d_{\min}}}\right\rfloor+1 \right\},$$
this summation can be written as
\begin{align}
 & \sum_{{\bf a}\in\mathbb{A}(\genGam,d_{\min};2N_t)}\frac{2\Nt\sqrt{\genGam}^{2\Nt-1}}{\|{\bf a}\|^{2\Nt-1}2^C\frac{2}{\sqrt{d_{\min}}}} \nonumber \\
 & \leq\sum^{\left\lfloor\sqrt{\frac{\genGam}{d_{\min}}}\right\rfloor}_{k=0}\sum_{k<\|{\bf a}\|\leq k+1}\frac{2\Nt\sqrt{\genGam}^{2\Nt-1}}{k^{2\Nt-1}2^C\frac{2}{\sqrt{d_{\min}}}}.
 \label{eq:96}
\end{align}

Denoting $\eta=\eta(N_t,\genGam,d_{\min},C)=\frac{2\Nt\sqrt{\genGam}^{2\Nt-1}}{2^C\frac{2}{\sqrt{d_{\min}}}}$, we apply Lemma 1 in \cite{OrdentlichErez2012} (a bound for the number of integer vectors contained in a ball of a given radius). Using this bound while noting that when $\|{\bf a}\|=1$ there are exactly $2\Nt$ integer vectors, the right hand side of (\ref{eq:96}) may be further bounded as
\begin{align}
&\leq \eta{\rm Vol}(\mathcal{B}_{2\Nt}(1)) \times  \left[2\Nt+\sum^{\left\lfloor\sqrt{\frac{\genGam}{d_{\min}}}\right\rfloor}_{k=1} \right.\nonumber\\
& \left. \left(\frac{\left(k+1+\frac{\sqrt{2\Nt}}{2}\right)^{2\Nt}-\left(\max\left(k-\frac{\sqrt{2\Nt}}{2} ,0\right)\right)^{2\Nt}}{k^{2\Nt-1}}\right)\right],
\label{eq:114}
\end{align}
where we note that  (\ref{eq:114}) trivially holds when
$\left\lfloor\sqrt{\frac{\genGam}{d_{\min}}}\right\rfloor=0$
since $\mathbb{A}(\genGam,d_{\min};2N_t)$ is the empty set in this case. Henceforth we assume that $\left\lfloor\sqrt{\frac{\genGam}{d_{\min}}}\right\rfloor\geq1$.
Further, the right hand side of (\ref{eq:114}) can be rewritten as
{
\begin{align}
& {\eta{\rm Vol}(\mathcal{B}_{2\Nt}(1))} \left[\underbrace{2\Nt}_{\rm I}+ \underbrace{\sum^{\left\lfloor\frac{\sqrt{2\Nt}}{2}\right\rfloor}_{k=1}\frac{\left(k+1+\frac{\sqrt{2\Nt}}{2}\right)^{2\Nt}}{k^{2\Nt-1}}}_{\rm II} + \right. \nonumber \\
& \left. \underbrace{\sum^{\left\lfloor\sqrt{\frac{\genGam}{d_{\min}}}\right\rfloor}_{k=\left\lfloor\frac{\sqrt{2\Nt}}{2}\right\rfloor+1}\frac{\left[\left(k+1+\frac{\sqrt{2\Nt}}{2}\right)^{2\Nt}-\left(k-\frac{\sqrt{2\Nt}}{2} \right)^{2\Nt}\right]}{k^{2\Nt-1}}}_{\rm III}
\right].
\label{eq:122}
\end{align}
}
We search for $c_1$ and $c_2$ (independent of $k$) such that
\begin{align}
\left(k+1+\frac{\sqrt{2\Nt}}{2}\right)^{2\Nt}\leq c_1 k^{2\Nt-1}
\label{eq:c1}
\end{align}
for $1\leq k \leq \left\lfloor\frac{\sqrt{2\Nt}}{2}\right\rfloor$, and
\begin{align}
\left[\left(k+1+\frac{\sqrt{2\Nt}}{2}\right)^{2\Nt}-\left(k-\frac{\sqrt{2\Nt}}{2} \right)^{2\Nt}\right]\leq c_2 k^{2\Nt-1}
\label{eq:c2}
\end{align}
for $k\geq1$, since it will then follow that
\begin{align}
II+III &\leq  \sum^{\left\lfloor\sqrt{\frac{\genGam}{d_{\min}}}\right\rfloor}_{k=1}  \max(c_1,c_2)
\label{eq:100}
\end{align}
We note that since (again assuming $\left\lfloor\sqrt{\frac{\genGam}{d_{\min}}}\right\rfloor\geq1$)
\begin{align}
2N_t \leq \sum^{\left\lfloor\sqrt{\frac{\genGam}{d_{\min}}}\right\rfloor}_{k=1} 2N_t,
\end{align}
it will thus further follow that
\begin{align}
I+II+III &\leq  \sum^{\left\lfloor\sqrt{\frac{\genGam}{d_{\min}}}\right\rfloor}_{k=1}  \left[2\Nt+\max(c_1,c_2)\right] \nonumber \\
&={\left\lfloor\sqrt{\frac{\genGam}{d_{\min}}}\right\rfloor} \left[2\Nt+\max(c_1,c_2)\right].
\label{eq:102}
\end{align}

To establish \eqref{eq:c1} and \eqref{eq:c2}, we first show that we may take
\begin{align}
c_1=\left(1+\sqrt{2\Nt}\right)^{2\Nt}.
\label{eq:126}
\end{align}
In other words, we need to establish that
\begin{align}
\frac{\left(k+1+\frac{\sqrt{2\Nt}}{2}\right)^{2\Nt}}{k^{2\Nt-1}}\leq \left(1+\sqrt{2\Nt}\right)^{2\Nt}
\end{align}
holds for $1\leq k \leq \left\lfloor\frac{\sqrt{2\Nt}}{2}\right\rfloor$.
Since $k\geq1$, we have
\begin{align}
 \frac{\left(k+1+\frac{\sqrt{2\Nt}}{2}\right)^{2\Nt}}{k^{2\Nt-1}} \leq \left(k+1+\frac{\sqrt{2\Nt}}{2}\right)^{2\Nt}.
\end{align}
Now, for $k\leq\frac{\sqrt{2\Nt}}{2}$, we have
\begin{align}
 \left(k+1+\frac{\sqrt{2\Nt}}{2}\right)^{2\Nt} \leq\left(1+\sqrt{2\Nt}\right)^{2\Nt}.
\end{align}
Hence, (\ref{eq:c1}) indeed holds with $c_1=\left(1+\sqrt{2\Nt}\right)^{2\Nt}$.

Next, we show that we may take
\begin{align}
c_2 =\left[\left(2+\frac{\sqrt{2\Nt}}{2}\right)^{2\Nt}-\left(1-\frac{\sqrt{2\Nt}}{2} \right)^{2\Nt}\right].
\label{eq:129}
\end{align}
Thus, we need to show that for $k\geq 1$, the following holds
\begin{align}
& \frac{1}{k^{2\Nt-1}}\left[\left(k+1+\frac{\sqrt{2\Nt}}{2}\right)^{2\Nt}-\left(k-\frac{\sqrt{2\Nt}}{2} \right)^{2\Nt}\right] \nonumber \\
& \leq \left[\left(2+\frac{\sqrt{2\Nt}}{2}\right)^{2\Nt}-\left(1-\frac{\sqrt{2\Nt}}{2} \right)^{2\Nt}\right].
\end{align}
Using the binomial expansion
\begin{align}
&\frac{1}{k^{2\Nt-1}}\left[\left(k+1+\frac{\sqrt{2\Nt}}{2}\right)^{2\Nt}-\left(k-\frac{\sqrt{2\Nt}}{2} \right)^{2\Nt}\right] \nonumber \\
&=\sum_{i=0}^{2\Nt}{2\Nt-i \choose i}k^{1-i}\left[\left(1+\frac{\sqrt{2\Nt}}{2}\right)^i-\left(-\frac{\sqrt{2\Nt}}{2} \right)^i\right]  \nonumber \\
&=\sum_{i=1}^{2\Nt}{2\Nt-i \choose i}k^{1-i}\left[\left(1+\frac{\sqrt{2\Nt}}{2}\right)^i-\left(-\frac{\sqrt{2\Nt}}{2} \right)^i\right]
\label{eq:132}\\
& \leq \sum_{i=1}^{2\Nt}{2\Nt-i \choose i}\left[\left(1+\frac{\sqrt{2\Nt}}{2}\right)^i-\left(-\frac{\sqrt{2\Nt}}{2} \right)^i\right]
\label{eq:133} \\
& = \left[\left(2+\frac{\sqrt{2\Nt}}{2}\right)^{2\Nt}-\left(1-\frac{\sqrt{2\Nt}}{2} \right)^{2\Nt}\right],
\end{align}
where (\ref{eq:133}) follows since each of the summands in (\ref{eq:132}) is monotonically decreasing in $k$ (when $k\geq1$ and $i \geq 1$). Thus, (\ref{eq:c2}) indeed holds when taking $c_2$ as defined in (\ref{eq:129}). Hence, we have established our choices for $c_1$ and $c_2$.

Now, since for $\Nt\geq2$ we have
\begin{align}
c_2&=\left[\left(2+\frac{\sqrt{2\Nt}}{2}\right)^{2\Nt}-\left(1-\frac{\sqrt{2\Nt}}{2} \right)^{2\Nt}\right] \nonumber \\
&\leq \left(2+\frac{\sqrt{2\Nt}}{2}\right)^{2\Nt}  \nonumber \\
&\leq \left(1+\sqrt{2\Nt}\right)^{2\Nt} \nonumber \\
&=c_1.
\end{align}
Recalling (\ref{eq:102}), it follows that
\begin{align}
I+II+III \leq {\left\lfloor\sqrt{\frac{\genGam}{d_{\min}}}\right\rfloor}\left(2\Nt+\left(1+\sqrt{2\Nt}\right)^{2\Nt}\right).
\label{eq:199}
\end{align}
 Applying (\ref{eq:199}) to (\ref{eq:122}), we get
{\small
\begin{align}
& \sum_{{\bf a}\in\mathbb{A}(\genGam,d_{\min};2N_t)}\frac{2\Nt\sqrt{\genGam}^{2\Nt-1}}{\|{\bf a}\|^{2\Nt-1}2^C\frac{2}{\sqrt{d_{\min}}}} \nonumber \\
& \leq {\left\lfloor\sqrt{\frac{\genGam}{d_{\min}}}\right\rfloor} \eta{\rm Vol}(\mathcal{B}_{2\Nt}(1))\cdot \left(2\Nt+\left(1+\sqrt{2\Nt}\right)^{2\Nt}\right)  \nonumber \\
& = {\left\lfloor\sqrt{\frac{\genGam}{d_{\min}}}\right\rfloor} \frac{2\left(2\Nt+\left(1+\sqrt{2\Nt}\right)^{2\Nt}\right)\Nt\sqrt{\genGam}^{2\Nt-1}}{2^C\frac{2}{\sqrt{d_{\min}}}}{\rm Vol}(\mathcal{B}_{2\Nt}(1))  \nonumber \\
& \leq \frac{2\left(2\Nt+\left(1+\sqrt{2\Nt}\right)^{2\Nt}\right)\Nt\sqrt{\genGam}^{2\Nt-1}}{2^C\frac{2}{\sqrt{d_{\min}}}}\sqrt{\frac{\genGam}{d_{\min}}}{\rm Vol}(\mathcal{B}_{2\Nt}(1))  \nonumber \\
& = \frac{\left(2\Nt+\left(1+\sqrt{2\Nt}\right)^{2\Nt}\right)\Nt{\genGam}^{\Nt}}{2^C}{\rm Vol}(\mathcal{B}_{2\Nt}(1)).
\end{align}
}
Further, setting $\genGam=2^{\frac{C-\Delta C}{\Nt}}\alNt$, we have
{\small
\begin{align}
 &\prob\left(R_{\rm IF}\left(\svv{D},\svv{V}\right)<C-\Delta C\right)
  \nonumber \\
 &\leq \frac{\left(2\Nt+\left(1+\sqrt{2\Nt}\right)^{2\Nt}\right)\Nt\left(2^{\frac{C-\Delta C}{\Nt}}\alNt\right)^{\Nt}}{2^C}{\rm Vol}(\mathcal{B}_{2\Nt}(1))  \nonumber \\
 &= \frac{\left(2\Nt+\left(1+\sqrt{2\Nt}\right)^{2\Nt}\right)\Nt2^{C-\Delta C}\alNt^{\Nt}}{2^C}{\rm Vol}(\mathcal{B}_{2\Nt}(1))  \nonumber \\
 &= \left(2\Nt+\left(1+\sqrt{2\Nt}\right)^{2\Nt}\right)\Nt \alNt^{\Nt}\frac{\pi^{\Nt}}{\Gamma(\Nt+1)}2^{-\Delta C}.
\end{align}
}
Hence,
\begin{align}
\prob\left(R_{\rm IF}\left(\svv{D},\svv{V}\right)<C-\Delta C\right) & \leq c(\Nt)2^{-\Delta C}
\end{align}
where
\begin{align}
c(\Nt)=\left(2\Nt+\left(1+\sqrt{2\Nt}\right)^{2\Nt}\right)\Nt \alNt^{\Nt}\frac{\pi^{\Nt}}{\Gamma(\Nt+1)}
\label{eq:appAlast}
\end{align}
is a constant that depends only on $\Nt$. We note that (\ref{eq:appAlast}) does not depend on $\svv{D}$ and hence  it holds also for the supremum over $\svv{D}\in \mathbb{D}(C;2N_t)$. Recalling~(\ref{eq:rewirte6withD}), we have
\begin{align}
P^{\rm WC}_{\rm out,IF}\left(C,\Delta C\right)&=\sup_{\svv{D}\in\mathbb{D}(C;2N_t)}\prob\left(R_{\rm IF}(\svv{D},\svv{V})<C-\Delta C\right)\nonumber\\
&\leq c(\Nt)2^{-\Delta C}
\end{align}
which concludes the proof.
\begin{remark}
It can be shown that the term $I$ (i.e., the term $2N_t$) in \eqref{eq:122} can be dropped from $c(N_t)$ and further that the choice $c_1=\left(2+\frac{\sqrt{2\Nt}}{2}\right)^{2\Nt}$ also satisfies (\ref{eq:c1}). Combining these two  observations, (\ref{eq:appAlast}) can be further tightened to
\begin{align}
c(\Nt)=\left(2+\frac{\sqrt{2\Nt}}{2}\right)^{2\Nt}\Nt \alNt^{\Nt}\frac{\pi^{\Nt}}{\Gamma(\Nt+1)}.
\end{align}
Figures~\ref{fig:lem1NoCountOpt}~and~\ref{fig:fig2} use this tighter bound.
\label{rem:private}
\end{remark}
\section{Tighter Bounds for $N_r\times2$ channels: Proof of Lemma~{\ref{lem:lem3} and Theorem~\ref{thm:thm2}}}
\label{sec:BoundFor2x2channels}

We consider the performance of an IF-SIC receiver over $N_r\times2$ channels. Thus, we now have $2N_t=4$.

As mentioned in Section~\ref{sec:precIFequalizationWithSic}, when using complex linear pre-processing matrices, the rates of both IF and IF-SIC come
in pairs. Denote
\begin{align}
        R^{1,\rm IF}(\svv{D},\svv{V})&=R_{1,\rm IF}(\svv{D},\svv{V})=R_{2,\rm IF}(\svv{D},\svv{V})
        \label{eq:119a}
\end{align}
and
\begin{align}
        R^{2,\rm IF}(\svv{D},\svv{V})&=R_{3,\rm IF}(\svv{D},\svv{V})=R_{4,\rm IF}(\svv{D},\svv{V}).
         \label{eq:119b}
\end{align}
where $R_{m,\rm IF}(\svv{D},\svv{V})$ is the rate of the $m$th equation (corresponding to the $m$th row of $\svv{A}$) as defined in \eqref{eq:comprate_m}, where we implicitly assume that $\svv{A}$ is the optimal matrix for IF. Similarly, denote
\begin{align}
        R^{1,\rm IF-SIC}(\svv{D},\svv{V})&=R_{1,\rm IF-SIC}(\svv{D},\svv{V})=R_{2,\rm IF-SIC}(\svv{D},\svv{V})
         \label{eq:119c}
\end{align}
and
\begin{align}
        R^{2,\rm IF-SIC}(\svv{D},\svv{V})&=R_{3,\rm IF-SIC}(\svv{D},\svv{V})=R_{4,\rm IF-SIC}(\svv{D},\svv{V}).
         \label{eq:119d}
\end{align}
We note that the (optimal) integer matrix $\svv{A}$ used for IF 
in \eqref{eq:119a}-\eqref{eq:119b} is  in general different than the (optimal) matrix $\svv{A}$ used for IF-SIC in \eqref{eq:119c}-\eqref{eq:119d}. 
Nonetheless, when applying IF-SIC, one decodes first the equation with the \emph{highest} SNR. Since for this equation SIC has no effect it follows that the first row of $\svv{A}$ is the same in both cases and hence
\begin{align}
R^{1,\rm IF-SIC}(\svv{D},\svv{V}) & = R^{1,\rm IF}(\svv{D},\svv{V}).
\label{eq:RifSIC1}
\end{align}
From Section III.A in \cite{OrdentlichErezNazer1:2013}, we have
\begin{align}
    \sum_{i=1}^{4}R_{i,\rm IF-SIC}(\svv{D},\svv{V})=C-\log|\det(\svv{A})|.
\end{align}
Furthermore, by Theorem 3 in \cite{OrdentlichErezNazer1:2013}, the optimal integer matrix $\svv{A}$ for IF-SIC is unimodular (i.e., has determinant 1 or -1). Hence,
\begin{align}
\sum_{i=1}^{4}R_{i,\rm IF-SIC}(\svv{D},\svv{V})=C.
\label{eq:C_sic}
\end{align}
Since we use IF-SIC with equal rate per stream, we have the following
\begin{align}
R_{\rm IF-SIC}(\svv{D},\svv{V})=4\min(R^{1,\rm IF-SIC}(\svv{D},\svv{V}),R^{2,\rm IF-SIC}(\svv{D},\svv{V})).
\label{eq:RifSIC}
\end{align}

Substituting (\ref{eq:RifSIC1}) into (\ref{eq:C_sic}), we have
\begin{align}
R^{2,\rm IF-SIC}(\svv{D},\svv{V})=\frac{C}{2}-R^{1,\rm IF}(\svv{D},\svv{V}).
\end{align}
Now, from Theorem 3 in \cite{ApproximateSumCapacity_OrdentlichErezNazer:2016} (with equivalent four real dimensions), we have
\begin{align}
2(R^{1,\rm IF}(\svv{D},\svv{V}) + R^{2,\rm IF}(\svv{D},\svv{V})) \geq C-4 \nonumber \\
R^{1,\rm IF}(\svv{D},\svv{V}) + R^{2,\rm IF}(\svv{D},\svv{V}) \geq \frac{C-4}{2}.
\end{align}
Since $R^{1,\rm IF}(\svv{D},\svv{V}) \geq R^{2,\rm IF}(\svv{D},\svv{V})$, it follows that
\begin{align}
R^{1,\rm IF}(\svv{D},\svv{V}) \geq \frac{C-4}{4}.
\label{eq:RifSICLowBound}
\end{align}
We conclude that
\begin{align}
R_{\rm IF-SIC}(\svv{D},\svv{V}) &= 4\min\left(R^{1,\rm IF-SIC}(\svv{D},\svv{V}),R^{2,\rm IF-SIC}(\svv{D},\svv{V})\right) \nonumber \\
& = 4\min\left(R^{1,\rm IF}(\svv{D},\svv{V}),\frac{C}{2}-R^{1,\rm IF}(\svv{D},\svv{V})\right) \nonumber \\
& \geq  4\min\left(\frac{C-4}{4},\frac{C}{2}-R^{1,\rm IF}(\svv{D},\svv{V})\right) \nonumber \\
& = \min\left({C-1},2C-4R^{1,\rm IF}(\svv{D},\svv{V})\right).
\end{align}

Henceforth, we analyze the outage
probability for $C>1$ and target rates that are no greater than $C-1$, so that the inequality $2C-4R^{1,\rm IF}(\svv{D},\svv{V})<C-1$ is satisfied. Thus, we consider gap-to-capacity values such that $\Delta C>1$. Our goal is to bound
\begin{align}
&\prob\left(R_{\rm IF-SIC}(\svv{D},\svv{V})<C-\Delta C\right)  \nonumber \\ &=\prob\left(2C-4\frac{1}{2}\log\left(\frac{1}{\lambda_1^2(\Lambda)}\right)<C-\Delta C\right)  \nonumber \\
&=\prob\left(-2\log\left(\frac{1}{\lambda_1^2(\Lambda)}\right)<-(C+\Delta C)\right)  \nonumber \\
&=\prob\left(\lambda_1^2(\Lambda)<2^{\nicefrac{-1}{2}(C+\Delta C)}\right).
\label{eq:PrLambda1Implicit-Comp}
\end{align}
We are now ready to prove Lemma~\ref{lem:lem3} and Theorem~\ref{thm:thm2}.
Let $\genGam=2^{\nicefrac{-1}{2}(C+\Delta C)}$. We wish to bound~(\ref{eq:PrLambda1Implicit-Comp}), or equivalently
\begin{align}
\prob\left(\lambda_1^2(\Lambda)<{\genGam}\right)=\prob\left(\lambda_1(\Lambda)<\sqrt{\genGam}\right).
\end{align}
for a given matrix $\svv{D}$. Note that the event $\lambda_1(\Lambda)<\sqrt{\genGam}$ is equivalent to the event
\begin{align}
\bigcup_{\mathbf{a}\in{\mathbb{Z}^{4}\setminus\{{\bf 0}\}}}||\svv{D}^{-1/2}\svv{V}^T\mathbf{a}||<\sqrt{\genGam}.
\end{align}
Applying the union bound yields
\begin{align}
\prob\left(\lambda_1(\Lambda)<\sqrt{\genGam}\right) \leq \sum_{{\bf a}\in\mathbb{Z}^{4}\setminus\{{\bf 0}\}}\prob\left(||\svv{D}^{-1/2}\svv{V}^T{\bf a}||<\sqrt{\genGam}\right).
\label{eq:outageProbReal2x2}
\end{align}
Note that if $\frac{||{\bf a}||}{\sqrt{d_{\max}}}>\sqrt{\genGam}$, we have
\begin{align}
    \prob\left(||\svv{D}^{-1/2}\svv{V}^T{\bf a}||<\sqrt{\genGam}\right)=0.
    \label{eq:173}
\end{align}
Therefore, using the notation of  (\ref{eq:A_beta_d}), the set of relevant vectors ${\bf a}$ is
\begin{align}
{\mathbb{A}(\genGam,1/d_{\max};4)}=\left\{{\bf a}\in \mathbb{Z}^{4}:0<||{\bf a}||<\sqrt{\genGam d_{\max}}\right\}.
\end{align}
It follows from (\ref{eq:outageProbReal2x2}) and (\ref{eq:173}) that
\begin{align}
&\prob\left(\lambda_{1}(\Lambda)<\sqrt{\genGam}\right)
\nonumber \\
&\leq \sum_{{\bf a}\in{\mathbb{A}(\genGam,1/d_{\max};4)}}\prob\left(\|\svv{D}^{-1/2}\svv{V}^T{\bf a}\|<\sqrt{\genGam}\right).
\end{align}

We now apply a similar derivation to that of   Section~\ref{sec:BoundforNrxNtchannels}. Applying Lemma~\ref{lem:lem1}, we have
\begin{align}
\prob\left(\|\svv{D}^{-1/2}\svv{V}^T{\bf a}\|<\sqrt{\genGam}\right)=\prob\left(\|\svv{D}^{-1/2}\svv{O}^T{\bf a}\|<\sqrt{\genGam}\right).
\end{align}
where $\svv{O}$ is drawn from the CRE. Hence, we can apply the same geometric interpretation as in Section~\ref{sec:BoundforNrxNtchannels} and interpret
$\prob\left(\|\svv{D}^{-1/2}\svv{O}^T{\bf a}\|<\sqrt{\genGam}\right)$ as the ratio of the surface area of the four-dimensional ellipsoid  inside a ball with radius $\sqrt{\genGam}$ and the surface area of this ellipsoid.
The axes of this ellipsoid are defined as
\begin{align}
x_i=\frac{\|{\bf a}\|}{\sqrt{d_i}}.
\end{align}
For the case of four real dimensions, (\ref{eq:ellipsCirecEq}) can be written as
\begin{align}
\sum_{{\bf a}\in{\mathbb{A}(\genGam,1/d_{\max};4)}}\prob\left(\|\svv{D}^{-1/2}{\bf o}_{\|{\bf a}\|}<\sqrt{\genGam}\right)=\nonumber \\
\sum_{{\bf a}\in{\mathbb{A}(\genGam,1/d_{\max};2N_t)}}\frac{\rm CAP_{\rm ell}}{L(x_1,x_2,x_3,x_4)}
\label{eq:ratioForSIC}
\end{align}
where
\begin{align}
{\rm CAP_{\rm ell}}<A_4(\sqrt{\genGam})=4\frac{\pi^2}{2}{\sqrt{\genGam}}^3\triangleq\overline{\rm CAP_{\rm ell}},
\label{eq:CAPellSIC}
\end{align}
and
\begin{align}
{L(x_1,x_2,x_3,x_4)}&>\frac{\pi^2}{2}\frac{\aNorm^4}{\prod_{i=1}^{4}\sqrt{d_i}}\left(\frac{2\sqrt{d_{\min}}}{\aNorm}+\frac{2\sqrt{d_{\max}}}{\aNorm}\right) \nonumber \\
& \geq {\pi^2}\frac{\aNorm^3}{2^C}\left({\sqrt{d_{\max}}}\right)\triangleq\underline{L}(x_1,x_2,x_3,x_4).
\label{eq:ellCirSIC}
\end{align}
Substituting~(\ref{eq:CAPellSIC}) and (\ref{eq:ellCirSIC}) in (\ref{eq:ratioForSIC}), we obtain
\begin{align}
&\sum_{{\bf a}\in{\mathbb{A}(\genGam,1/d_{\max};4)}}\frac{\rm CAP_{\rm ell}}{L(x_1,x_2,x_3,x_4)}
\nonumber \\
&<\sum_{{\bf a}\in{\mathbb{A}(\genGam,1/d_{\max};4)}}\frac{\overline{\rm CAP_{\rm ell}}}{\underline{L}(x_1,x_2,x_3,x_4)} \nonumber \\
&=\sum_{{\bf a}\in{\mathbb{A}(\genGam,1/d_{\max};4)}}\frac{2\pi^2{\sqrt{\genGam}}^3}{{\pi^2}\frac{\aNorm^3}{2^C}\left({\sqrt{d_{\max}}}\right)}.
\end{align}
Recalling that $\genGam=2^{\nicefrac{-1}{2}(C+\Delta C)}$, we get that for $\Delta C<1$
\begin{align}
&\prob\left(R_{\rm IF-SIC}(\svv{D},\svv{V})<C-\Delta C\right) \nonumber \\
&\leq \sum_{{\bf a}\in\mathbb{A}(\genGam,1/d_{\max};4)}\frac{2\pi^2{2^{\nicefrac{-3}{4}(C+\Delta C)}}}{{\pi^2}\frac{\aNorm^3}{2^C}\left({\sqrt{d_{\max}}}\right)},
\end{align}
which proves Lemma~\ref{lem:lem3}.

To establish Theorem~\ref{thm:thm2}, we follow the footsteps of the proof of Theorem~\ref{thm:thm1} (noting that now $2N_t=4$) to obtain
\begin{align}
&\sum_{{\bf a}\in{\mathbb{A}(\genGam,1/d_{\max};4)}}\prob\left(\|\svv{D}^{-1/2}{\bf o}_{\|{\bf a}\|}<\sqrt{\genGam}\right) \nonumber \\
& \leq \sum_{{\bf a}\in{\mathbb{A}(\genGam,1/d_{\max};4)}}\frac{2\pi^2{\sqrt{\genGam}}^3}{{\pi^2}\frac{\aNorm^3}{2^C}{\sqrt{d_{\max}}}}  \nonumber \\
& = \sum_{{\bf a}\in{\mathbb{A}(\genGam,1/d_{\max};4)}}\frac{2{\sqrt{\genGam}}^3{2^C}}{\aNorm^3{\sqrt{d_{\max}}}}  \nonumber \\
& = \frac{2{\sqrt{\genGam}}^3{2^C}}{{\sqrt{d_{\max}}}}  \sum^{\left\lfloor\sqrt{\genGam d_{\max}}\right\rfloor}_{k=0}
\sum_{k<\|{\bf a}\|\leq k+1} \frac{1}{k^3} \nonumber \\
& \leq  \frac{2{\sqrt{\genGam}}^3{2^C}}{{\sqrt{d_{\max}}}} \left[2N_t+  \sum^{\left\lfloor\sqrt{\genGam d_{\max}}\right\rfloor}_{k=1} \frac{\left(k+2\right)^{4}-\left(k-1\right)^{4}}{k^3} \right] \nonumber \\
& \leq  \frac{2{\sqrt{\genGam}}^3{2^C}}{{\sqrt{d_{\max}}}} \left[4+  \sum^{\left\lfloor\sqrt{\genGam d_{\max}}\right\rfloor}_{k=1}\frac{81k^3}{k^3} \right]  \nonumber \\
& = \sum^{\left\lfloor\sqrt{\genGam d_{\max}}\right\rfloor}_{k=1}\frac{{\sqrt{\genGam}}^3{2^C}\pi^2 85}{{\sqrt{d_{\max}}}} \nonumber \\
& \leq \frac{{\sqrt{\genGam}}^3{2^C}\pi^2 85 \sqrt{\genGam d_{\max}}}{{\sqrt{d_{\max}}}}  \nonumber \\
& \leq {\genGam}^2 {2^C}\pi^2 85.
\label{eq:thm2NoD}
\end{align}
As (\ref{eq:thm2NoD}) does not depend on $\svv{D}$,  it follows that the bound holds also for the supremum over $\svv{D} \in \mathbb{D}(C;4)$. Now since $\genGam=2^{-\nicefrac{1}{2}(C+\Delta C)}$, we get
\begin{align}
P^{\rm WC}_{\rm out,IF-SIC}\left(C,\Delta C\right)&\leq85\cdot\pi^2\cdot2^C2^{-C-\Delta C} \nonumber \\
&=85\cdot\pi^2\cdot2^{-\Delta C}.
\label{eq:148}
\end{align}
\begin{remark}
\label{rem:rem7}
Similar to Remark~\ref{rem:private} in Appendix~\ref{sec:proofOfTheorem1}, (\ref{eq:148}) can be further tightened to
\begin{align}
P^{\rm WC}_{\rm out,IF-SIC}\left(C,\Delta C\right) \leq 81\cdot\pi^2\cdot2^{-\Delta C}.
\end{align}
\end{remark}
\end{appendices}

\bibliographystyle{IEEEtran}
\bibliography{eladd}
 \begin{IEEEbiographynophoto}{Elad Domanovitz}
 received the B.Sc. degree (cum laude) and the M.Sc. degree in 2005 and 2011, respectively, in electrical engineering
 from Tel Aviv University, Israel. He is currently working toward the Ph.D. degree at Tel Aviv University.
 \end{IEEEbiographynophoto}
 \begin{IEEEbiographynophoto}{Uri Erez}
 (M'09) was born in Tel-Aviv, Israel, on October 27, 1971.
 He received the B.Sc. degree in mathematics and physics and the M.Sc. and Ph.D. degrees in electrical engineering from Tel-Aviv University in 1996, 1999, and 2003, respectively. During 2003-2004, he was a Postdoctoral
 Associate at the Signals, Information and Algorithms Laboratory at the Massachusetts Institute of Technology (MIT), Cambridge. Since 2005, he has been with the Department of Electrical Engineering-Systems at Tel-Aviv
 University. His research interests are in the general areas of information theory and digital communication. He served in the years 2009-2011 as Associate Editor for Coding Techniques for the IEEE TRANSACTIONS ON INFORMATION THEORY.\end{IEEEbiographynophoto}

\end{document}

%% file: eli_macros.tex
\DeclareMathOperator{\Tr}{Tr}

%% file: rv_defs.tex
\DeclareMathAlphabet{\mathbsf}{OT1}{cmss}{bx}{n}
\DeclareMathAlphabet{\mathssf}{OT1}{cmss}{m}{sl}
\DeclareMathAlphabet{\mathcsf}{OT1}{cmss}{sbc}{n}

\newcommand{\svv}[1]{\mathbf{#1}}

\DeclareSymbolFont{bsfletters}{OT1}{cmss}{bx}{n}  
\DeclareSymbolFont{ssfletters}{OT1}{cmss}{m}{n}
\DeclareMathSymbol{\bsfGamma}{0}{bsfletters}{'000}
\DeclareMathSymbol{\ssfGamma}{0}{ssfletters}{'000}
\DeclareMathSymbol{\bsfDelta}{0}{bsfletters}{'001}
\DeclareMathSymbol{\ssfDelta}{0}{ssfletters}{'001}
\DeclareMathSymbol{\bsfTheta}{0}{bsfletters}{'002}
\DeclareMathSymbol{\ssfTheta}{0}{ssfletters}{'002}
\DeclareMathSymbol{\bsfLambda}{0}{bsfletters}{'003}
\DeclareMathSymbol{\ssfLambda}{0}{ssfletters}{'003}
\DeclareMathSymbol{\bsfXi}{0}{bsfletters}{'004}
\DeclareMathSymbol{\ssfXi}{0}{ssfletters}{'004}
\DeclareMathSymbol{\bsfPi}{0}{bsfletters}{'005}
\DeclareMathSymbol{\ssfPi}{0}{ssfletters}{'005}
\DeclareMathSymbol{\bsfSigma}{0}{bsfletters}{'006}
\DeclareMathSymbol{\ssfSigma}{0}{ssfletters}{'006}
\DeclareMathSymbol{\bsfUpsilon}{0}{bsfletters}{'007}
\DeclareMathSymbol{\ssfUpsilon}{0}{ssfletters}{'007}
\DeclareMathSymbol{\bsfPhi}{0}{bsfletters}{'010}
\DeclareMathSymbol{\ssfPhi}{0}{ssfletters}{'010}
\DeclareMathSymbol{\bsfPsi}{0}{bsfletters}{'011}
\DeclareMathSymbol{\ssfPsi}{0}{ssfletters}{'011}
\DeclareMathSymbol{\bsfOmega}{0}{bsfletters}{'012}
\DeclareMathSymbol{\ssfOmega}{0}{ssfletters}{'012}